\documentclass[sigconf, nonacm]{acmart}

\usepackage[english]{babel}
\usepackage[disable]{todonotes}
\usepackage{semantic}
\usepackage{cleveref}

\usepackage[noamsthm,nohypref,nofnttls,chgbar]{fmocdmac}

\AtEndPreamble {

}
\usepackage{tikz}
\usetikzlibrary{arrows,arrows.meta,calc}
\usetikzlibrary{shadows}
\usetikzlibrary{shapes.multipart}
\usetikzlibrary{positioning}
\usetikzlibrary{shadows}
\usetikzlibrary{shapes.multipart}
\usepackage{xcolor-solarized}
\newtheorem{definition}{Definition}

\newtheorem{proposition}{Proposition}
\newtheorem{lemma}{Lemma}
\newtheorem{claim}{Claim}

\newcommand{\graph}{\mathcal{G}}
\newcommand{\schema}{\mathcal{S}}

\newcommand{\Nodes}{\mathcal{N}}
\newcommand{\Values}{\mathcal{V}}

\newcommand{\Predicates}{\mathcal{P}}
\newcommand{\Keys}{\mathcal{K}}
\newcommand{\Records}{\mathcal{R}}
\newcommand{\nodes}{\mathsf{Nodes}}

\newcommand{\keys}{\mathsf{Keys}}
\newcommand{\values}{\mathsf{Values}}

\newcommand{\sel}{\mathit{sel}}

\newcommand{\IRIs}{\mathsf{IRIs}}
\newcommand{\Blanks}{\mathsf{Blanks}}
\newcommand{\Literals}{\mathsf{Literals}}

\newcommand{\Exkey}[1]{\mathit{#1}}
\newcommand{\Exprop}[1]{\mathsf{#1}}
\newcommand{\Exvt}[1]{\mathbbm{#1}}

\newcommand{\exowns}{\Exprop{ownsAccount}}
\newcommand{\exaccess}{\Exprop{hasAcccess}}
\newcommand{\exinvited}{\Exprop{invited}}
\newcommand{\exemail}{\Exkey{email}}
\newcommand{\excard}{\Exkey{card}}
\newcommand{\exprivileged}{\Exkey{privileged}}

\newcommand{\OMIT}[1]{}

\newcommand{\shapeTerm}{SHACL shape\xspace}
\newcommand{\selTerm}{SHACL selector\xspace}
\newcommand{\SHACLSchema}{\schema}
\newcommand{\pathExpr}{\pi}

\newcommand{\id}{\mathsf{id}}

\newcommand{\eq}{\mathsf{eq}}
\newcommand{\disj}{\mathsf{disj}}
\newcommand{\geqn}[2]{\exists^{\geq #1}#2.}
\newcommand{\leqn}[2]{\exists^{\leq #1}#2.}

\newcommand{\hasvalue}{\mathsf{test}}
\newcommand{\test}{\mathsf{test}}
\newcommand{\closed}{\mathsf{closed}}

\newcommand{\iexpr}[2]{\llbracket #1 \rrbracket^{#2}}

\usepackage{bbm}

\newcommand{\sem}[1]{\llbracket{#1}\rrbracket}
\newcommand{\gsem}[1]{\llbracket{#1}\rrbracket^{\graph}}

\newcommand{\rp}[1]{\texttt{\{} #1 \texttt{\}}}  % record pattern
\newcommand{\closedRT}[1]{\rp{#1}}  % closed record type
\newcommand{\emptyRec}{\textbf{r}_{\emptyset}}

\newcommand{\tOr}{\mathbin{\texttt{|}}}
\newcommand{\tAnd}{\mathbin{\texttt{\&}}}

\newcommand{\ValueTypes}{\mathcal{T}} % the set of value types
\newcommand{\vtype}{\mathbbm{v}} % a value type variable
\newcommand{\ntype}{\mathbbm{c}} % a node type variable
\newcommand{\etype}{\mathbbm{e}} % an edge type variable

\newcommand{\gDef}{{\color{orange} \ \Coloneqq \ }}
\newcommand{\gMid}{{\color{orange} \ \big|\ }}
\newcommand{\gEnd}{{\color{orange} \ .\ }}

\newcommand{\pexpr}{\pi}
\newcommand{\ppexpr}{\bar{\pi}}

\newcommand{\et}[3]{{#1}\stackrel{#2}{\rightarrow}{#3}}  % (primitive) edge type
\newcommand{\pwc}{\star}  % predicate wildcard
\newcommand{\Key}{\textsf{\bf Key}}

\newcommand{\keyIsVal}[2]{[{#1}={#2}]}

\newcommand{\shexschema}{\schema}

\newcommand{\shexsel}{\mathit{sel}}

\newcommand{\shexneigh}[1]{\left\{ #1 \right\}}
\newcommand{\shexneighopen}[1]{\left\{ #1 \right\}^{\circ}}
\newcommand{\shexallte}{\top}
\newcommand{\shextop}{\shexneigh{\shexallte}}

\newcommand{\se}{\mathit{\varphi}}
\newcommand{\te}{\mathit{e}}
\newcommand{\tte}{e}

\newcommand{\shextest}{\mathsf{test}}
\newcommand{\shexhasvalue}{\mathsf{test}}
\newcommand{\shexeach}{\mathop{;}}
\newcommand{\shexone}{\mathop{|}}
\newcommand{\shexinverse}[1]{#1^{-}}

\newcommand{\shexneg}[1]{\neg #1}
\newcommand{\shexneginv}[1]{\neg{\shexinverse{#1}}}

\newcommand{\ttopen}{\textit{op}_{\pm}}
\newcommand{\ttclosed}{\textit{op}_{-}}

\newcommand{\copyswap}{\mathit{copyswap}}

\newcommand{\neigh}{\mathsf{Neigh}}

\newcommand{\Triples}{\mathcal{E}}
\newcommand{\any}{\mathbbm{any}}

\newcommand{\ognjen}[1]{\todo[inline]{Ognjen: #1}}

\hypersetup {
  pdftitle  = {Common Foundations for SHACL, ShEx, and PG-Schema},
  pdfauthor = {S. Ahmetaj, I. Boneva, J. Hidders, K. Hose, M. Jakubowski,
    J.E. Labra-Gayo, W. Martens, F. Mogavero, F. Murlak, C. Okulmus,
    A. Polleres, O. Savkovic, M. Simkus, D. Tomaszuk}
}

\copyrightyear{2025}
\acmYear{2025}
\setcopyright{cc}
\setcctype{by}
\acmConference[WWW '25]{Proceedings of the ACM Web Conference 2025}{April 28-May 2, 2025}{Sydney, NSW, Australia}
\acmBooktitle{Proceedings of the ACM Web Conference 2025 (WWW '25), April 28-May 2, 2025, Sydney, NSW, Australia}
\acmDOI{10.1145/3696410.3714694}
\acmISBN{979-8-4007-1274-6/25/04}

\begin{document}

\title{Common Foundations for SHACL, ShEx, and PG-Schema}

\author{Shqiponja Ahmetaj}
\orcid{0000-0003-3165-3568}
\email{shqiponja.ahmetaj@tuwien.ac.at}
\affiliation{
  \institution{TU Wien} \city{Vienna}\country{Austria}
}

\author{Iovka Boneva}
\orcid{0000-0002-2696-7303}
\email{iovka.boneva@univ-lille.fr}
\affiliation{%
  \institution{Univ. Lille, CNRS, Inria, Centrale Lille, UMR 9189 CRIStAL}
  \postcode{F-59000}\city{Lille}\country{France}
}

\author{Jan Hidders}
\orcid{0000-0002-8865-4329}
\email{j.hidders@bbk.ac.uk}
\affiliation{
  \institution{Birkbeck, University of London} \city{London}\country{UK}
}

\author{Katja Hose}
\orcid{0000-0001-7025-8099}
\email{katja.hose@tuwien.ac.at}
\affiliation{
  \institution{TU Wien} \city{Vienna}\country{Austria}
}

\author{Maxime Jaku{\-}bow{\-}ski}
\orcid{0000-0002-7420-1337}
\email{maxime.jakubowski@tuwien.ac.at}
\affiliation{
  \institution{TU Wien} \city{Vienna}\country{Austria}
}

\author{Jose-Emilio Labra-Gayo}
\orcid{0000-0001-8907-5348}
\email{labra@uniovi.es}
\affiliation{
  \institution{University of Oviedo} \city{Oviedo}\country{Spain}
}

\author{Wim Martens}
\orcid{0000-0001-9480-3522}
\email{wim.martens@uni-bayreuth.de}
\affiliation{
  \institution{University of Bayreuth} \city{Bayreuth}\country{Germany}
}

\author{Fabio Mogavero}
\orcid{0000-0002-5140-5783}
\email{fabio.mogavero@unina.it}
\affiliation{
  \institution{Universit\`a di Napoli Federico II} \city{Naples}\country{Italy}
}

\author{Filip Mur{\-}lak}
\orcid{0000-0003-0989-3717}
\email{f.murlak@uw.edu.pl}
\affiliation{
  \institution{University of Warsaw} \city{Warsaw}\country{Poland}
}

\author{Cem Okulmus}
\orcid{0000-0002-7742-0439}
\email{cem.okulmus@uni-paderborn.de}
\affiliation{
  \institution{Paderborn University} \city{Paderborn}\country{Germany}
}

\author{Axel Polleres}
\orcid{0000-0001-5670-1146}
\email{axel.polleres@wu.ac.at}
\affiliation{
  \institution{WU Wien} \city{Vienna}\country{Austria} \\
  \institution{CSH Vienna} \city{Vienna}\country{Austria}
}

\author{Ognjen Savkovi\'c}
\orcid{0000-0002-9141-3008}
\email{ognjen.savkovic@unibz.it}
\affiliation{
  \institution{Free University of Bolzano} \city{Bolzano}\country{Italy}
}

\author{Mantas \v{S}imkus}
\orcid{0000-0003-0632-0294}
\email{mantas.simkus@tuwien.ac.at}
\affiliation{
  \institution{TU Wien} \city{Vienna}\country{Austria}
}

\author{Dominik Tomaszuk}
\orcid{0000-0003-1806-067X}
\email{d.tomaszuk@uwb.edu.pl}
\affiliation{
  \institution{TU Wien} \city{Vienna}\country{Austria} \\
  \institution{University of Bialystok} \city{Bialystok}\country{Poland}
}

\renewcommand{\shortauthors}{Ahmetaj, et al.}

% Begin of file Abstract.tex

\begin{abstract}
Graphs have emerged as an important foundation for a variety of applications,
including capturing and reasoning over factual knowledge, semantic data
integration, social networks, and providing factual knowledge for machine
learning algorithms.
To formalise certain properties of the data and to ensure data quality, there is
a need to describe the \emph{schema} of such graphs.
Because of the breadth of applications and availability of different data
models, such as RDF and property graphs, both the Semantic Web and the database
community have independently developed \emph{graph schema languages}: SHACL,
ShEx, and PG-Schema.
Each language has its unique approach to defining constraints and validating
graph data, leaving potential users in the dark about their commonalities and
differences.
In this paper, we provide formal, concise definitions of the \emph{core
components} of each of these schema languages.
We employ a uniform framework to facilitate a comprehensive comparison between
the languages and identify a common set of functionalities, shedding light on
both overlapping and distinctive features of the three languages.
\end{abstract}

% End of file Abstract.tex

\maketitle

% Begin of file 1-Introduction.tex

\section{Introduction}

Driven by the unprecedented growth of interconnected data, \emph{graph-based
data representations} have emerged as an expressive and versatile framework for
modelling and analysing connections in data sets~\cite{SakrBVIAAAABBDV21}.
This rapid growth however, has led to a proliferation of diverse approaches,
each with its own identity and perspective.

The two most prominent graph data models are \emph{RDF} (Resource Description
Framework)~\cite{CWL14} and \emph{Property Graphs}~\cite{BFVY23}.
In RDF, data is modelled as a collection of triples, each consisting of a
subject, predicate, and object.
Such triples naturally represent either edges in a directed labelled graph
(where the predicates represent relationships between nodes), or
attributes-value pairs of nodes.
That is, objects can both be entities or atomic (literal) values.
In contrast, Property Graphs model data as nodes and edges, where both can have
labels and records attached, allowing for a flexible representation of
attributes directly on the entities and relationships.

Similarly to the different data models, we are also seeing different approaches
towards \emph{schema languages} for graph-structured data.
Traditionally, in the Semantic Web community, schema and constraint languages
have been \emph{descriptive}, focusing on flexibility to accommodate varying
structures.
However, there has been a growing need for more \emph{prescriptive} schemas
that focus on \emph{validation of data}.
At the same time, in the Database community, schemas have traditionally been
prescriptive but, since the rise of semi-structured data, the demand for
descriptive schemas has been growing.
Thus, the philosophies of schemas in the two communities have been growing
closer together.

For RDF, there are two main schema languages: SHACL (Shapes Constraint
Language)~\cite{KK17}, which is also a W3C recommendation, and ShEx (Shape
Expressions)~\cite{PGS14}.
In the realm of Property Graphs, the current main approach is
PG-Schema~\cite{ABDF23,ABDF21}; it was developed with liaisons to the GQL and SQL/PGQ standardization committees and is currently being used as a basis for extending these standards.
%, since it is being used as a basis by the GQL and SQL/PGQ ISO committees.}
The development processes of these languages have been quite different.
For SHACL and ShEx, the formal semantics were only introduced after their
initial implementations, echoing the evolution of programming languages.
Indeed, an analysis of SHACL's expressive power and associated decision problems
appeared in the literature~\cite{LSRLS20,PKMN20,PK21,BJB22,PKM22,BJVdB24} only
after it was published as a W3C recommendation, leading up to a fully
\emph{recursive variant} of the
language~\cite{CRS18,CFRS19,ACORSS20,BJ21,PKM22}, whose semantics had been left
undefined in the standard.
A similar scenario occurred with ShEx, where formal analyses were only conducted
in later phases~\cite{BGP17,SBG15}.
PG-Schema developed in the opposite direction.
Here, a group of experts from industry and academia first defined the main ideas
in a sequence of research papers~\cite{ABDF21,ABDF23} and the implementation is
expected to follow.

Since these three languages have been developed in different communities, in the
course of different processes, it is no surprise that they are quite different.
SHACL, ShEx, and PG-Schema use an array of diverse approaches for defining how
their components work, ranging from \emph{declarative} (formulae that
\emph{specify what to look for}) to \emph{generative} (expressions that
\emph{generate the matching content}), and even combinations thereof.
The bottom line is that we are left with three approaches to express a ``schema
for graph-structured data'' that are very different at first glance.

As a group of authors coming from both the Semantic Web and Database
communities, we believe that there is a \emph{need for common understanding}.
While the functionalities of schemas and constraints used in the two communities
largely overlap, it is a daunting task to understand the essence of languages,
such as SHACL, ShEx, and PG-Schema.
In this paper, we therefore aim to shed light on the common aspects and the
differences between these three languages.
We focus on non-recursive schemas, as neither PG-Schema nor
standard SHACL support recursion and also in the academic community the
discussion on the semantics of recursive SHACL has not reached consensus
yet~\cite{CRS18,CFRS19,ACORSS20,BJ21,PKM22,OS24}.

Using a common framework, we provide crisp definitions of the main aspects of
the languages.
Since the languages operate on different data models, as a first step we
introduce the \emph{Common Graph Data Model}, a mathematical representation of
data that \emph{canonically embeds} both RDF graphs and Property Graphs (see
Section~\ref{sec:prl}, which also develops general common foundations).
Precise abstractions of the languages themselves are presented in
Sections~\ref{sec:shacl} (SHACL),~\ref{sec:shex} (ShEx),
and~\ref{sec:pgschema-simplified} (PG-Schema);
\todo{Remove or reformulate!}
in the Appendices we explain how and why we sometimes deviate from the original
formalisms.
Each of these sections contains examples to give readers an immediate intuition
about what kinds of conditions each language can express.
Then, in Section~\ref{sec:core}, we present the \emph{Common Graph Schema
Language (CoGSL)}, which consists of functionalities shared by them all.

Casting all three languages in a common framework has the immediate advantage
that the reader can identify common functionalities \emph{based on the syntax
only}: on the one hand, we aim at giving the same semantics to schema language
components that syntactically look the same, and on the other hand, we can
provide examples of properties that distinguish the three languages using simple
syntactic constructs that are not part of the common core.
Aside from corner cases, properties expressed using constructs outside the
common core are generally not expressible in all three languages.
By providing an understanding of fundamental differences and similarities
between the three schema languages, we hope to benefit both practitioners in
choosing a schema language fitting their needs, and researchers in studying the
complexity and expressiveness of schema languages.

% End of file 1-Introduction.tex

%!TEX root = Article.tex

% Begin of file 2-Preliminaries.tex

\section{Foundations}
\label{sec:prl}

In this section we present some material that we will need in the subsequent
sections, and define a data model that consists of common aspects of RDF and
Property Graphs.

\subsection{A Common Data Model}

When developing a common framework for SHACL, ShEx, and PG-Schema, the first
challenge is establishing  a \emph{common data model}, since SHACL and ShEx work
on RDF, whereas PG-Schema works on Property Graphs.
Rather than using a model that generalises  both RDF and Property Graphs, we
propose a simple model, called \emph{common graphs}, which we obtained by asking
what, fundamentally, are the \emph{common aspects} of RDF and Property Graphs
(Appendix~\ref{sec:appendix-foundations} gives more details on the distilling of
common graphs).

Let us assume disjoint countable sets of nodes $\Nodes$, values $\Values$,
predicates $\Predicates$, and keys $\Keys$ (sometimes called properties).

% We sometimes say \emph{element} for a node or a value, and \emph{label} for a predicate or key. \todo{Drop if not used.}

\begin{definition}
  A \emph{common graph} is a pair $\graph = (E, \rho)$ where
  \begin{itemize}[\textbullet]
  \item
    $E \subseteq_{\mathit{fin}} \Nodes \times \Predicates \times \Nodes$ is its
    set of edges (which carry predicates), and
  \item
    $\rho \colon \Nodes \times \Keys \pto \Values$ is a finite-domain partial
    function mapping node-key pairs to values.
  \end{itemize}
  The set of nodes of a common graph $\graph$, written $\nodes(\graph)$,
  consists of all elements of $\Nodes$ that occur in $E$ or in the domain of
  $\rho$.
  Similarly, $\keys(\graph)$ is the subset of $\Keys$ that is used in $\rho$,
  and $\values(\graph)$ is the subset of $\Values$ that is used in $\rho$ (that
  is, the range of $\rho$).
\end{definition}

% \begin{example}[Media Service Common Graph] \label{ex:sharedScenario}
% To illustrate the common graphs, we introduce the following scenario. We assume a data model that has users, who can access and own accounts and invite other users to their accounts. Users have keys, such as email and credit-card. An example for this can be seen in~\Cref{fig}.
% % The nodes correspond to  conceptual classes, which will be identified by their available properties and keys. Properties are depicted as directed arrows, and keys are shown inside the conceptual classes.
% % The boxes inform about the available categories of nodes, with the keys they may have available (such as the key $\Exkey{plan}$ for nodes of category ``Account''), and properties connect nodes via directed arrows (such as $\Exprop{buyer}$, which connects nodes of category ``Sale'' and ``Account'').
% \end{example}

\begin{example}
  \label{ex:common-graph}
  Consider Figure~\ref{fig:common-graph}, containing a graph to store
  information about \emph{users} who may have access to (possibly multiple)
  \emph{accounts} in, \eg, a media streaming service.
  In this example, we have six nodes describing four persons ($u_1,...,u_4$) and
  two accounts ($a_1$, $a_2$).
  As a common graph $\graph = (E, \rho)$, the nodes are $a_1$, $u_1$, etc.
  Examples of edges in $E$ are $(u_2, \exaccess,a_1)$ and $(u_3, \exinvited,
  u_2)$.
  Furthermore, we have $\rho(u_2, \exemail) =$ d@d.d and $\rho(a_1,card) =
  1234$.
  So, $E$ captures the arrows in the figure (labelled with predicates) and
  $\rho$ captures the key/value information for each node.
  %
% Moreover, 3 predicates are used, appearing in Figure~\ref{fig:common-graph} as labels on links between nodes, representing the relation~$E$. Nodes are further associated with some key-value pairs, representing the function $\rho$.
  %
  Notice that a person may be the owner of an account, and may potentially have
  access to other accounts.
  This is captured using the predicates $\exowns$ and $\exaccess$, respectively.
  In addition, the system implements an invitation functionality, where users
  may invite other people to join the platform.
  The previous invitations are recorded using the predicate $\exinvited$.
  Both accounts and users may be privileged, which is stored via a Boolean value
  of the key~$\exprivileged$.
  We note that the presence of the key $\exemail$ (\resp, of the key (credit)
  $\excard$) is associated with, and indeed identifies users (\resp, accounts).
\end{example}

% \todo[inline]{In the example, worth noting that the graph node names are names, and not identities. Maybe it would be better to name them A, B, C, D to avoid misunderstanding?}

\begin{figure}[t]
\resizebox{1\linewidth}{!}{
  \includegraphics{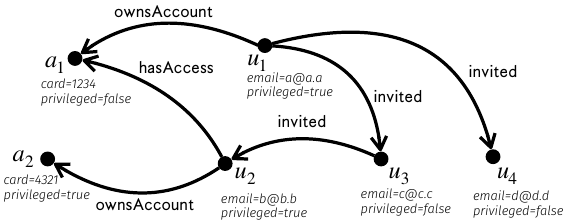}
}
\Description{A diagram of the user common graph.}
\caption{The media service common graph. }
\label{fig:common-graph}
\end{figure}

It is easy to see that every common graph is a property graph (as per the formal
definition of property graphs~\cite{ABDF23}).
A common graph can also be seen as a set of triples, as in RDF.
Let
\[
  \Triples
=
  \left( \Nodes \times \Predicates \times \Nodes \right)
\;\cup\;
  \left( \Nodes \times \Keys \times \Values \right)\,.
\]
Then, a common graph can be seen as a finite set $\graph \subseteq \Triples$
such that for each $u \in \Nodes$ and $k \in \Keys$ there is at most one
$v \in \Values$ such that $(u, k, v) \in \graph$.
Indeed, a common graph $(E, \rho)$ corresponds to
\[
  E \;\cup\; \{ (u, k, v) \mid \rho(u,k) = v\}\;.
\]
When we write $\rho(u, k) = v$ we assume that $\rho$ is defined on $(u, k)$.

\medskip

\noindent\emph{Throughout the paper we see property graph $\graph$
simultaneously as a pair $(E, \rho)$ and as a set of triples from $\Triples$,
switching between these perspectives depending on what is most convenient at a
given moment.}

\subsection{Node Contents and  Neighbourhoods}

Let $\Records$ be the set of all \emph{records}, \ie, finite-domain partial
functions $r \colon \Keys \pto \Values$.
We write records as sets of pairs $\left\{ (k_1, w_1), \dots (k_n, w_n)
\right\}$ where $k_1, \dots, k_n$ are all different, meaning that $k_i$ is
mapped to $w_i$.

For a common graph $\graph = (E,\rho)$ and node $v$ in $\graph$, by a slight
abuse of notation we write $\rho(v)$ for the record $\left\{ (k, w) \mid
\rho(v,k) = w \right\}$ that collects all key-value pairs associated with node
$v$ in $\graph$.
We call $\rho(v)$ the \emph{content} of node $v$ in $\graph$.
This is how PG-Schema interprets common graphs: it views key-value pairs in
$\rho(v)$ as \emph{properties} of the node $v$, rather than independent,
navigable objects in the graph.

SHACL and ShEx, on the other hand, view common graphs as sets of triples and
make little distinction between keys and predicates.
The following notion---when applied to a node---uniformly captures the local
context of this node from that perspective: the content of the node and all
edges incident with the node.

%\begin{definition}[Neighbourhood]
%Given a common graph $\graph = (E,\rho)$ and a node $v\in\Nodes$, we write $\neigh_\graph(v)$ for the common graph $(E',\rho')$ where $E' = \left \{ (u_1,p,u_2) \in  E \mid u_1 = v \text{ or } u_2 = v\right\}$ and $\rho'$ is obtained by restricting $\rho$ so that $\rho'(v) = \rho(v)$ and $\rho'(u)$ is empty for all $u\neq v$. Similarly, for $w\in\Values$, we let $\neigh_\graph(w)$ be the common graph $(\emptyset,\rho')$ where $\rho'(u) = \left\{(k,w')\in\rho(u)\mid w'=w\right\}$ for all $u\in\Nodes$.
%Given a common graph $\graph$ and a node or value $v\in\Nodes\cup\Values$, the \emph{neighbourhood of $v$ in $\graph$}, written $\neigh_\graph(v)$, is the common graph consisting of triples $(u_1, p, u_2)$ from $\graph$ such that $p\in\Predicates\cup\Keys$ and either $u_1=v$ or $u_2=v$.
%\end{definition}

%That is, for $v\in\Nodes$,  $\neigh_\graph(v)$ is a star-shaped graph where only the central node has non-empty content.  For $w\in\Values$, $\neigh_\graph(w)$ is a graph with no edges and only a single value occurring in the contents of nodes.

%If we view common graphs as sets of triples, $\neigh_\graph(v)$ for $v\in\Nodes\cup\Values$ is simply the set of all triples from $\graph$ that mention $v$.

%We will also use the notion of \emph{partial neighbourhoods}, where only specified subsets of keys and predicates are taken into account.

%It is easiest to define it seeing common graphs as sets of triples.

\begin{definition}[Neighbourhood]
  Given a common graph $\graph$ and a node or value $v \in \Nodes \cup \Values$,
  the \emph{neighbourhood} of $v$ in $\graph$ is $\neigh_\graph(v) = \left\{
  (u_1, p, u_2) \in \graph \mid u_1 = v \text{ or } u_2 = v \right\}$.
  %
% \todo[inline]{Wim: This is ill-defined. We do say before that a common graph can be viewed as a set of triples if we want to think about it as RDF. But this definition should also apply to the PG view. We should be clearer about what we mean with the key/value pairs and only use ingredients from Def 1. In fact, if we take the RDF view, the definition is inconsistent with text below that says that, if $v$ is a value, then the neighborhood has no edges.}
% \todo[inline]{Suggestion to rephrase: introduce $\graph = (E,\rho)$ and say $\neigh_\graph(v) = \{(u_1,p,u_2) \in E \mid ... \} \cup \{???\}$ (Actually I don't understand yet what we want wrt $\rho$.)}
% \todo[inline]{Filip: In many places in the paper we treat $\graph$ as a pair $(E,\rho)$ or as a subset of $\Triples$, whatever is more convenient. It should suffice to warn the reader that we do this. We could write the definition in terms of $(E,\rho)$, but it would be clumsy. I really think it is fine as written.  On the other hand, if this is not helping, we can probably just skip this definition entirely and introduce only the $\pm$ variant of neighbourhoods in the section on ShEx.}
% \todo[inline]{Wim: OK, I understand better now what's intended and clarified below.}
\end{definition}

\todo{JH: Is this actually used anywhere?}

When $v \in \Nodes$, then $\neigh_\graph(v)$ is a star-shaped graph
where only the central node has non-empty content.
When $v \in \Values$, then $\neigh_\graph(v)$ consists of all the nodes in
$\graph$ that have some key with value $v$, which is a common graph with no
edges and a restricted function $\rho$.

%\todo[inline]{Maybe move to respective sections. Could also save space.}

\subsection{Value Types}

We assume an enumerable set of \emph{value types} $\ValueTypes$.
The reader should think of value types as \texttt{integer}, \texttt{boolean},
\texttt{date}, \etc
Formally, for each value type $\vtype \in \ValueTypes$, we assume that there is
a set $\sem{\vtype} \subseteq \Values$ of all values of that type and that each
value $v \in \Values$ belongs to some type, \ie, there is at least one $\vtype
\in \ValueTypes$ such that $v \in \sem{\vtype}$.
Finally, we assume that there is a type $\any \in \ValueTypes$ such that
$\sem{\any} = \Values$.

\subsection{Shapes and Schemas}
\label{ssec:shapes}

We formulate all three schema languages using \emph{shapes}, which are unary
formulas describing the graph's structure around a \emph{focus} node or a value.
Shapes will be expressed in different formalisms, specific to the schema
language; for each of these formalisms we will define when a focus node or value
$v \in \Nodes \cup \Values$ \emph{satisfies} shape $\varphi$ in a common graph
$\graph$, written $\graph, v \models \varphi$.

Inspired by ShEx \emph{shape maps}, we abstract a schema $\schema$ as a set of
pairs $(\sel,\varphi)$, where $\varphi$ is a shape and $\sel$ is a
\emph{selector}.
A selector is also a shape, but usually a very simple one, typically checking
the presence of an incident edge with a given predicate, or a property with a
given key.
A graph $\graph$ is \emph{valid} \wrt $\schema$, in symbols $\graph \models
\schema$, if
\[
  \graph, v \models \sel
\quad \text{implies} \quad
  \graph, v \models \varphi,
\]
for all $v \in \Nodes \cup \Values$ and $(\mathit{sel}, \varphi) \in \schema$.
That is, for each focus node or value satisfying the selector, the graph around
it looks as specified by the shape.
We call schemas $\schema$ and $\schema'$ \emph{equivalent} if $\graph \models
\schema$ \iff $\graph \models \schema'$, for all $\graph$.
In what follows, we may use $\mathit{sel} \Rightarrow \varphi$ to indicate a
pair $(\mathit{sel}, \varphi)$ from a schema $\SHACLSchema$.

% \begin{example}[Schemas over Media Service Common Graph]
%     \label{ex:ShapeExample}

% We stay in the same scenario introduced in \Cref{ex:sharedScenario}. We list here illustrative examples for requirements on common graphs that can be imposed via schemas.  To give an intuitive idea about the selector and the shape, we indicate this informally by splitting the sentences into an initial part that selects nodes or values, and the second part which must hold for these elements:\\
% \noindent
% \emph{For every account}, there must exist a primary credit card ; \\
% \noindent \emph{For every account}, there are  five users of it or less;\\
% \emph{Every owner of an account}, has a unique email address.
% \end{example}

\begin{example}
  \label{ex:constraint-desc}
  We next describe some constraints one may want to express in the domain of
  Example~\ref{ex:common-graph}.
  \begin{enumerate}[(C1)]
  \item
    We may want the values associated to certain keys to belong to concrete
    datatypes, like strings or Boolean values.
    In our example, we want to state that the value of the key $\excard$ is
    always an integer.
  \item
    We may expect the existence of a value associated to a key, an outgoing
    edge, or even a complex path for a given source node.
    For our example, we require that all owners of an account have an email
    address defined.
  \item
    We may want to express database-like uniqueness constraints.
    For instance, we may wish to ensure that the email address of an account
    owner uniquely identifies them.
  \item
    We may want to ensure that all paths of a certain kind end in nodes with
    some desired properties. For example, if an account is privileged, then all
    users that have access to it should also be privileged.
  \item
    We may want to put an upper bound on the number of nodes reached from a
    given node by certain paths. For instance, every user may have access to at
    most 5 accounts.
\end{enumerate}

% \todo[inline]{Wim: Reminder to self. I'd like to illustrate some open/closed things here. (There's no time anymore for this.)}
% \todo[inline]{Wim: More urgently though, we should explain better about how we model things. Let's say that ``users'' are those nodes that have an email key and ``accounts'' are those that have a card key?}
% \todo[inline]{Iovka: I support the need to make this precise. Then, should we use these two selectors in all examples?\\
% Also, we might say that we need this trick because we do not have rdf:type nor labels on nodes.}
% \todo[inline]{Cem: After discussion with Filip, I fixed the setting such that it is keys that identify users and accounts. Problem: this makes C2 awkward. }

\end{example}

% End of file 2-Preliminaries.tex

%!TEX root = Article.tex

% Begin of file 3-SHACL.tex

\section{SHACL on common graphs}
\label{sec:shacl}

\begin{table}[t]
  \caption{Evaluation of a path expressions.}
  \label{tab:seme2}
  \centering
  \begin{tabular}{cl}
    \toprule
    $\pathExpr$ & $\iexpr{\pathExpr}{\graph} \subseteq (\Nodes\cup\Values)\times(\Nodes \times \Values)\ $ \\ % for $\graph = (E, \rho)$\\
    \midrule
    $\id$ & $\{ (v,v) \mid v \in \Nodes  \cup \Values \}$\\[2pt]
    $q$ & $\{(v,u)\mid (v,q,u)\in \graph \}$ \\[2pt]
    %$k$ & $\{(v,u)\mid \rho(v,k)=u\}$ \\[2pt]
    $\pathExpr^{-}$ & $\{(v,u)\mid (u,v) \in \iexpr{\pathExpr}{\graph}\}$ \\[2pt]
    $\pathExpr \cdot \pathExpr'$ & $\{(v,u) \mid \exists v': (v,v')\in\iexpr{\pathExpr}{\graph}\land (v',u)\in\iexpr{\pathExpr'}{\graph}\}$ \\[2pt]
    $\pathExpr\cup \pathExpr'$ & $\iexpr{\pathExpr}{\graph}\cup\iexpr{\pathExpr'}{\graph}$\\[2pt]
   $\pathExpr^{*}$ & $ \iexpr{\id}{\graph} \cup \iexpr{\pathExpr}{\graph}  \cup \iexpr{\pathExpr \cdot \pathExpr}{\graph} \cup \ldots $ \\[2pt]
    %$\pathExpr^{*}$ & $ \iexpr{\id}{\graph} \cup \{ (a,c) \mid (a,b) \in \iexpr{\pathExpr}{\graph} \land (b,c) \in \iexpr{\pathExpr^{*}}{\graph} \} $\\
       % the transitive reflexive closure of $\iexpr{\pathExpr}{\graph}$\\
    \bottomrule
  \end{tabular}
\end{table}

\newcommand{\defs}{\mathit{def}}
\renewcommand{\models}{\vDash}
\newcommand{\nmodels}{\nvDash}

We first treat SHACL, because it is conceptually the simplest of the three
languages.
It is essentially a logic---some call it a \emph{description logic in
disguise}~\cite{BJB22}.
Our abstraction is inspired by~\cite{MJPHD}.
We focus on the standard, non-recursive SHACL, leaving recursive extensions
\cite{CRS18,ACORSS20,BJ21,PKM22,OS24} for the future.
Some features of SHACL are incompatible with common graphs, and are therefore
omitted (see Appendix~\ref{app:standard-shacl}).
\todo{Remove reference to the appendix.}

%It defines shapes with formulas and relies on a generative formalism only to specify paths in the graph.

% \todo[inline]{Maybe we should mention SHACL core and also mention briefly the features not considered here and refer to the appendix for a more extensive discussion?}

%\todo[inline]{Filip: I do not know what SHACL core is. If we talk about it, we should explain. Do we have to?}

\begin{definition}[Path Expression]
  A \emph{path expression} $\pathExpr$ is given by the following grammar:
  \[
    \pathExpr
  \gDef
        \id
  \gMid q
  \gMid \pathExpr^{-}
  \gMid \pathExpr \cdot \pathExpr
  \gMid \pathExpr \cup \pathExpr
  \gMid \pathExpr^{*}
  \gEnd
  \]
  with $q \in \Predicates \cup \Keys$ and $\id$ the identity relation (or empty
  word).
%\todo[inline]{JH: The font for $\id$ is different from other keywords such as $\hasvalue$ and $\test$. Is there a reason for this?}
\end{definition}

\begin{definition}[SHACL Shape]
  \label{def:shacl-shape}
  A \emph{SHACL shape} $\varphi$ is given by the following grammar:
  \begin{align*}
    \varphi
  \gDef \
  & \top
  \gMid \hasvalue(c)
  \gMid \test(\vtype)
  \gMid \closed(Q)
  \gMid \eq(\pathExpr, p)
  \gMid \\
  & \disj(\pathExpr, p)
  \gMid \neg \varphi
  \gMid \varphi \land \varphi
  \gMid \varphi \lor \varphi
  \gMid \geqn{n}{\pathExpr}{\varphi}
  \gMid \leqn{n}{\pathExpr}{\varphi}
  \gEnd
  \end{align*}
  %with $c\in\Nodes\cup\Values$,
  with $c \in \Values$, $\vtype \in \ValueTypes$, $Q \subseteq_{\mathit{fin}}
  \Predicates \cup \Keys$, $p \in \Predicates$, and $n$ a natural number.
  We may use $\exists \pathExpr \ldotp \varphi$ as syntactic sugar for
  $\exists^{\geq 1} \pathExpr \ldotp \varphi$.
\end{definition}

\begin{definition}[SHACL Selector]
  A \selTerm $\mathit{sel}$ is a SHACL shape of a restricted form, given by the
  following grammar:
%\[ \mathit{sel} \gDef \geqn{1}{p}{\top}  \gMid \geqn{1}{k}{\top}  \gMid \geqn{1}{p^{-}}{\top} \gMid  \geqn{1}{k^{-}}{\top} \gMid  \hasvalue(c)\gEnd \]
  \[
    \mathit{sel}
  \gDef
        \exists\, q \ldotp \top
  \gMid \exists\, q^{-} \ldotp \top
  \gMid \hasvalue(c)
  \gEnd
  \]
  with $q \in \Predicates \cup \Keys$, and
%$c \in \Nodes \cup \Values$.
  $c \in \Values$.
\end{definition}

Putting it together, a \emph{SHACL Schema} $\SHACLSchema$ is a finite set of
pairs $(\mathit{sel}, \varphi)$, where $\mathit{sel}$ is a \selTerm and
$\varphi$ is a \shapeTerm.

To define the semantics of SHACL schemas, we first define in
Table~\ref{tab:seme2} the semantics of a SHACL path expression $\pi$ on a graph
$\graph$ as a binary relation $\iexpr{\pathExpr}{\graph}$ over $\Nodes \cup
\Values$.
The semantics of SHACL shapes is defined in Table~\ref{tab:semphi2}, which
specifies when a node or value $v$ \emph{satisfies} a \shapeTerm  $\varphi$ \wrt
a $\graph$, written $\graph, v \models \varphi$.
Note that both $\iexpr{\pathExpr}{\graph}$ and $\{ v \in \Nodes \cup \Values
\mid \graph, v \models \varphi\}$ may be infinite: for example,
$\sem{\text{id}}^\graph$ is the identity relation over the infinite set $\Nodes
\cup \Values$.

The semantics of SHACL schemas then follows \Cref{ssec:shapes}.
Importantly, SHACL selectors always select a finite subset of $\Nodes \cup
\Values$: the selected nodes or values come either from the selector itself, in
the case of $\hasvalue(c)$, or from $\graph$, in the remaining four cases.
For example, $\exists p \ldotp \top$  selects those nodes of $\graph$ that have
an outgoing $p$-edge in $\graph$---it is grounded to $\graph$ in the second line
of Table~\ref{tab:seme2}.
In consequence, each pair $(\mathit{sel}, \varphi)$ in a SHACL schema tests the
inclusion of a finite set of nodes or values in a possibly infinite set.
%\todo[inline]{Added explanation about finiteness vs infinity. Before, we said nothing at all about it, which I think was bad. Maybe my explanation can still be improved though!}

\begin{table}[t]
  \caption{Semantics of a \shapeTerm $\varphi$ .}
  \label{tab:semphi2}
  \centering
  \begin{tabular}{cl}
    \toprule
    $\varphi$ & $\graph, v \models \varphi$ if: \\
    \midrule

    $\top$ & trivially satisfied \\[2pt]
        $\hasvalue(c)$ & $v = c$\\[2pt]
    $\test(\vtype)$ & $v \in \sem{\vtype}$\\[2pt]

    $\closed(Q)$ &
                 % \begin{aligned}[t]
                    $  \forall p \in (\Predicates\cup\Keys) \setminus Q:$ not $\graph, v \models \exists^{\geq 1} p. \top$
                      %\notexists b \Rightarrow (y\in Q \; \land \\
                 %     & \qquad \text{if $\exists k: \rho(a,k)$ is defined, then $k\in Q)$}
                 %   \end{aligned}
                 %   $
                 \\[2pt]

    % $\hasshape(s)$ & $\SHACLShapeDef, \graph, v \models \SHACLShapeDef(s)$ \\
    % $\closed(Q)$ & $
    %                \begin{aligned}[t]
    %                  &\forall (x,y,z)\in E: x = v \Rightarrow (y\in Q \; \land \\
    %                  & \qquad \text{if $\exists k: \rho(a,k)$ is defined, then $k\in Q)$}
    %                \end{aligned}
    %                $\\

    $\eq(\pathExpr,p)$ & $\{ u \mid (v,u) \in \iexpr{\pathExpr}{\graph} \} = \{ u \mid (v,u) \in \iexpr{p}{\graph} \} $ \\[2pt]
    $\disj(\pathExpr,p)$ & $\{ u \mid (v,u) \in \iexpr{\pathExpr}{\graph} \} \cap \{ u \mid (v,u) \in \iexpr{p}{\graph} \} = \emptyset$\\[2pt]
    $\neg \varphi $ & not $\graph, v \models \varphi$  \\[2pt]
    $\varphi \land \varphi' $ & $\graph, v \models \varphi$ and $\graph, v  \models \varphi'$  \\[2pt]
    $\varphi \lor \varphi' $  & $\graph, v \models \varphi$ or $\graph, v  \models \varphi'$\\[2pt]

    % $\lessthan(F,E)$ & $b<c$ for all $b\in\iexpr{E}{\graph}(v)$ and $c\in\iexpr{F}{\graph}(v)$\\
    % $\lessthaneq(F,E)$ & $b\leq c$ for all $b\in\iexpr{E}{\graph}(v)$ and $c\in\iexpr{F}{\graph}(v)$\\
    % $\uniquelang(E)$ & $b$ and $c$ have different language tags \\ & for $b,c\in\iexpr{E}{\graph}(v)$ \\
    % $\forall E.\varphi$ & every $b\in\iexpr{E}{\graph}(v)$ satisfies $\graph,b\models\varphi$\\
    $\geqn{n}{\pathExpr}{\varphi}$ & $\#\{u\mid (v,u) \in \iexpr{\pathExpr}{\graph} \land \graph,u\models\varphi
    \}\geq n$\\[2pt]
    $\leqn{n}{\pathExpr}{\varphi}$ & $\#\{u\mid (v,u) \in \iexpr{\pathExpr}{\graph} \land\graph,u\models\varphi \}\leq n$\\
    \bottomrule
  \end{tabular}
\end{table}

\begin{example}

For better readability we write $\exists \pi$ instead of $\exists^{\ge 1} \pi
\ldotp \top$ (that is, we omit $\top$) and $\forall \pathExpr \ldotp \varphi$
instead of $\leqn{0}{\pathExpr}{\lnot \varphi}$.
% $\exists^{=n}\pi.\varphi$ instead of $\exists^{\leq n}\pi.\varphi \land \exists^{\geq n}\pi.\varphi$
%$\exists \pi.\varphi$ for $\exists^{\ge 1}\pi.\varphi$ and
% $\ge_n \pi.\varphi$ instead of $\exists^{\ge n}\pi.\varphi$ (and similarly for $\leq$) and
% $\exists \pi$ instead of $\exists^{\ge 1}\pi.\top$ (that is, we omit $\top$) and $\forall \pathExpr. \varphi$ instead of $\leqn{0}{\pathExpr}{\lnot \varphi}$. %We shall also often omit $\top$ to improve readability, so we write $\exists \pi$ to express $\exists^{\ge 1}\pi.\top$.
% \noindent
Let us see how the constraints from Example~\ref{ex:constraint-desc} can be
handled in SHACL.
For (C1), we assume the value type $\mathbbm{int}$ with the obvious meaning.
The  following SHACL constraints express the constraints (C1--C5):
\begin{align*}
  & \exists \excard^{-} \Rightarrow \test(\mathbbm{int})
    & \mbox{(C1)} \\
  & \exists \exowns     \Rightarrow \exists \exemail
    & \mbox{(C2)} \\
  & \exists\exemail^{-} \Rightarrow \exists^{\leq 1} \exemail^{-}
    & \mbox{(C3)} \\
  & \exists \excard      \Rightarrow (\exists \exprivileged \ldotp \neg
      \hasvalue(\mathit{true})) \,\lor
    & \\
  & \qquad  \forall \exaccess^{-} \ldotp (\exists \exprivileged \ldotp
      \hasvalue(\mathit{true}))
    & \mbox{(C4)} \\
% & \exists\Exprop{ownsAccount} \Rightarrow \exists \Exkey{email}  & \\
  & \exists \exemail \Rightarrow \leqn{5}{\exaccess}
    & \mbox{(C5)}
\end{align*}
Concerning constraint (C3), notice that by using inverse \textit{email} edges,
the constraint indeed states that the email addresses uniquely identify users.
\end{example}

The constructs $\eq(\pathExpr, p)$ and $\disj(\pathExpr, p)$ are unique to
SHACL.
Let us see them in use.

\begin{example}\label{ex:fancy-shacl-eq}
  Using $\eq(\pathExpr, p)$, we can say, for instance, that an owner of an
  account also has access to it:
  \[
    \exists \exowns \Rightarrow  \eq(\exaccess \cup \exowns, \exaccess)\,.
  \]
  Note how we use $\eq$ and $\cup$ to express that the existence of one path
  ($\exowns$) implies the existence of another path ($\exaccess$) with the same
  endpoints.
\end{example}

A key feature in SHACL that is not available in ShEx is the ability to use
regular expressions to talk about complex paths.
%\todo[inline]{ShEx has recursion in general, so for paths with existentials (without counting) it may be possible to express them by nesting of shapes.}
%\todo[inline]{Filip: Yup. We talk about it after defining the core. But the sentence above is still true: there are no regular expressions for paths in ShEx. I rephrased a bit.}
This provides a limited, still non-trivial, form of recursive navigation in the
graph, even though the standard SHACL does not support recursive constraints
(in contrast to standard ShEx).
% See below for an example.

\begin{example}
  \label{ex:fancy-shacl-paths}
  Suppose that in~\Cref{fig:common-graph}, we impose that for every node with a
  $\exprivileged$ key, either its value is $\mathit{false}$ or, along inverse
  $\exinvited$ edges there is a unique, privileged ``ancestor'', which has no
  further inverse  $\exinvited$ edges.
% we wanted to express that a privileged user may only invite other privileged users, who in turn can also only invite other privileged users.
  This is expressible as follows:
  \begin{align*}
    & \exists \exprivileged \Rightarrow \exists \exprivileged \ldotp
      \hasvalue(\mathit{false}) \lor \\
    & \quad \exists^{\leq 1} {\exinvited^{-}}^{*} \ldotp \big( \exists
      \exprivileged \ldotp \hasvalue(\mathit{true}) \land \exists^{\leq 0}
      \exinvited^{-}  \big) \,.
  \end{align*}
\end{example}

\section{ShEx on common graphs}
\label{sec:shex}

While SHACL is conceptually the simplest of the three languages, ShEx lies at the opposite end of the spectrum. It is an intricate, nested combination of a simple logic for shapes and a powerful formalism (triple expressions) for generating the allowed neighbourhoods. In this work we focus on non-recursive ShEx, where shapes and triple expressions can be nested multiple times, but cannot be recursive.
%\todo[inline]{Ognjen: i believe this is another kind of recursion, and at this point is a bit confusing.}
%\todo[inline]{Filip: Good catch. I replaced "mutually recursive" with "nested". Now it is about the same kind of recursion. }
This allows us to simplify the abstraction without compromising our primary goal of understanding the common features, as neither PG-Schema nor standard SHACL support such a general recursion mechanism.
The abstraction of ShEx over common graphs is based on the treatment of ShEx on RDF triples~\cite{BGP17}. Deviations from standard ShEx are discussed in Appendix~\ref{app:shex-appendix}.

\begin{definition}[shapes and triple expressions]
\label{def:shex-syntax}
ShEx \emph{shapes} $\se$ and
\emph{closed triple expressions} $e$
are defined by the  grammar
\begin{align*}
\se  \gDef & \shexhasvalue(c)
    \gMid
    \ \shextest(\text{$\vtype$})
    \gMid\! \shexneigh{\tte \shexeach \ttclosed }
    \gMid\! \shexneigh{\tte \shexeach \ttopen }
    \gMid\! \se \land \se
    \gMid\! \se \lor \se
    \gMid\! \lnot \se
    \gEnd \\
\tte \gDef &
    \ \varepsilon
    \ \ \gMid \ \
    q.\se
    \ \ \gMid \ \
    \shexinverse q.\se
    \ \ \gMid \ \
    \tte \shexeach \tte
    \ \ \gMid \ \
    \tte \shexone \tte
    \ \ \gMid \ \
    \tte^{*}
  \gEnd \\
  % pushinng op at the end to simplify the parsing
  \ttclosed \gDef &\ (\shexneginv{R})^{*} \gEnd \\
\ttopen \gDef &\  (\shexneginv{R})^{*} \shexeach (\shexneg{Q})^{*}  \gEnd
\end{align*}
where $c \in\Values$,
$\vtype\in\ValueTypes$, $q \in \Predicates \cup \Keys$, and $R,Q\subseteq_{\mathit{fin}}\Predicates \cup \Keys$. We refer to expressions derived from $\tte \shexeach \ttclosed$ and $\tte \shexeach \ttopen$ as \emph{half-open} and \emph{open triple expressions}, respectively.
\end{definition}

The notion of satisfaction for ShEx shapes and the semantics of triple expressions are defined by mutual recursion in Table~\ref{tab:semantics-shex-shape-expr} and Table~\ref{tab:semantics-shex-triple-expr-denotational}.
Triple expressions are used to specify neighbourhoods of nodes and values. They require to consider incoming and outgoing edges separately. For this purpose we decorate incoming edges with ${}^{-}$. Formally, we introduce a fresh predicate $p^{-}$ for each $p\in\Predicates$ and a fresh key $k^{-}$ for each $k\in\Keys$. We  let $\Predicates^{-}= \left\{p^{-} \mid p\in \Predicates\right\}$,
$\Keys^{-}= \left\{k^{-} \mid k\in \Keys\right\}$,
$\Triples^{-} = \Nodes \times \Predicates^{-} \times \Nodes \cup \Values \times \Keys^{-} \times \Nodes$, and define $\neigh^{\pm}_\graph(v) \subseteq \Triples \cup\Triples^{-}$ as
\[ \left \{ (v, p, v') \mid (v, p, v') \in \graph\right\} \cup \left\{(v, p^{-}, v') \mid (v', p, v) \in \graph \right\}.\] Compared to $\neigh_\graph(v)$, apart from flipping the incoming edges and marking them with ${}^{-}$, we also represent each loop $(v,p,v)$ twice: once as an outgoing edge $(v,p,v)$ and once as an incoming edge $(v,p^{-},v)$.
In Table~\ref{tab:semantics-shex-triple-expr-denotational}, we treat $\shexneg{Q}$ and  $\shexneginv{R}$ as triple expressions. So, the rule for $e^{*}$ gives semantics to $(\shexneg{Q})^{*}$ and  $(\shexneginv{R})^{*}$, and the rule for $e_1\shexeach e_2$ gives semantics to open and half-open triple expressions. In Table~\ref{tab:semantics-shex-shape-expr}, $f$ is an  open or half-open triple expression.

Closed triple expressions $\tte$ define neighbourhoods that use only a finite number of predicates and keys (also called \emph{closed} in ShEx terminology) and cannot be directly used in shape expressions.
Half-open triple expressions $\te\shexeach (\shexneginv{R})^{*}$ allow any \emph{incoming} triples whose predicate or key is not in  $R$. Open triple expressions $\te\shexeach (\shexneginv{R})^{*} \shexeach (\shexneg{Q})^{*}$ additionally allow any \emph{outgoing} triples whose predicate or key is not in $Q$.
Let $\shexallte = \varepsilon \shexeach (\shexneginv{\emptyset})^{*} \shexeach {(\shexneg{\emptyset})}^{*}$. Then $\shexallte$ describes all possible neighbourhoods, and $\shextop$ is satisfied in every node and in every value of every graph.

\begin{example}
The ShEx shape $\shexneigh{p.\varphi_1 \shexeach p.\varphi_2;\shexallte}$ specifies nodes with at least two different $p$-successors, one satisfying $\varphi_1$ and one satisfying $\varphi_2$. Note that this is different from SHACL shape $\exists p. \varphi_1 \land \exists p.\varphi_2$ which says that the node has a $p$-successor satisfying $\varphi_1$ and a $p$-successor satisfying $\varphi_2$, but they might not be different.
\end{example}

\begin{example}
    Assume that integers and strings are represented by $\mathbbm{int},\mathbbm{str}\in \ValueTypes$.
    The ShEx shape
    \[ \shexneigh{\exemail.\shextest(\mathbbm{str}) \,\shexeach\, (\excard.\shextest(\mathbbm{int}) \,\shexone\, \varepsilon) \,\shexeach\, (\shexneginv{\emptyset})^{*}}
    \]
    specifies nodes with an $\exemail$ property with a string value, an optional $\excard$ property with an integer value, arbitrary incoming edges, and no other properties or outgoing edges.
    To allow additional properties and outgoing edges, we replace $ (\shexneginv{\emptyset})^{*}$ with $\shexallte$.
    The modified shape can be rewritten using $\land$ as
    \[ \shexneigh{\exemail.\shextest(\mathbbm{str}) \,\shexeach\,\shexallte} \land \shexneigh{(\excard.\shextest(\mathbbm{int}) \,\shexone\, \varepsilon) \,\shexeach\, \shexallte}
    \] but the original shape cannot be rewritten in a similar way.
\end{example}

%Note also that the $\shexeach$ and $\shexone$ operators are associative and commutative.

\begin{definition}[ShEx Selectors]
\label{def:shex-selector}
A ShEx selector is a ShEx shape  of a restricted form, defined by the grammar
\begin{align*}
\shexsel \gDef & \shexhasvalue(c)
    \gMid \shexneigh{q.\shexhasvalue(c) \shexeach \shexallte}
    %\gMid \shexneigh{\shexinverse p.\shexhasvalue(c) \shexeach \shexallte}
    \gMid %\\
    \shexneigh{q.\shextop \shexeach \shexallte}
    \gMid \shexneigh{\shexinverse q.\shextop \shexeach \shexallte}
    \gEnd
\end{align*}
where $q \in \Predicates \cup \Keys$ and
%$c \in \Nodes \cup \Values$.
$c \in \Values$.
\end{definition}

Following Section~\ref{ssec:shapes}, a \emph{ShEx schema} $\shexschema$ is a set of pairs of the form $(\shexsel, \se)$ where $\se$ is a ShEx shape and $\shexsel$ is a ShEx selector.

\begin{table}
  \caption{Satisfaction of ShEx shapes.}
  \label{tab:semantics-shex-shape-expr}
  \centering
  \begin{tabular}{cl}
    \toprule
    $\se$ & $\graph, v \models \se\ $ for $v\in\Nodes\cup\Values$ \\
    \midrule
    $\shexhasvalue(c)$ & $v = c$ \\
    $\shextest(\vtype)$ & $v \in \sem{\vtype}$ \\
    $\shexneigh{f}$ & $\neigh^\pm_\graph(v) \in \sem{f}_v^\graph$ \\
    $\se_1 \land \se_2$ & $\graph, v \models \se_1$ and $\graph, v \models \se_2$ \\
    $\se_1 \lor \se_2$ & $\graph, v \models \se_1$ or $\graph, v \models \se_2$ \\
    $\lnot \se$ & not $\graph, v \models \se$ \\
    \bottomrule
  \end{tabular}
\end{table}

\begin{table}
  \caption{Semantics of triple expressions.}
  \label{tab:semantics-shex-triple-expr-denotational}
  \centering
  \begin{tabular}{cl}
    \toprule
    $\te$ & $\sem{\te}_v^\graph \subseteq 2^{\Triples \cup \Triples^{-}}$ \\[2pt]
    \midrule
    $\varepsilon$ & $\{\emptyset\}$ \\[2pt]
    $q.\se$ & $\big\{\{(v, q, v')\} \subseteq \Triples \ \big|\  \graph, v' \models \varphi\big\}$ \\[2pt]
    $\shexinverse q.\se$ & $\big\{\{(v, \shexinverse q, v')\} \subseteq \shexinverse\Triples\ \big|\ \graph, v' \models \varphi\big\}$  \\[2pt]
    $\te_1 \shexeach \te_2$ & $\left\{ T_1 \cup T_2 \ \middle|\  T_1\in \sem{\te_1}_v^\graph\,,\  T_2\in \sem{\te_2}_v^\graph\,,\ T_1\cap T_2 = \emptyset \right\}$ \\[2pt]
    $\te_1 \shexone \te_2$ & $\sem{\te_1}_v^\graph \cup \sem{\te_2}_v^\graph$ \\[2pt]
    $\te^{*}$
    & $ \{\emptyset\} \cup \bigcup_{n=1}^\infty \left \{\, T_1 \cup \dots \cup T_n \ \bigg| \begin{array}{l}
    T_1, \dots, T_n \in \sem{\te}_v^\graph \text{ and } \\
    %T_1, \dots, T_n \text{ are pairwise disjoint}
    T_i\cap T_j = \emptyset \text{ for all } i\neq j
    \end{array} \!
    \right\}$ \\[2pt]
    $\shexneg{Q}$ & $\big\{ \{(v, q , v')\} \subseteq  \Triples\ \big|\  q\notin Q\big\}$ \\[2pt]
    $\shexneginv{R}$ & $\big\{ \{(v, \shexinverse q , v')\} \subseteq  \shexinverse\Triples\ \big|\  q\notin R\big\}$
    \\
    \bottomrule
  \end{tabular}
\end{table}
%\ognjen{In table 4, why not just $(v, q, v')\in \Triples$?}
%%%% Because we want sets of neighbourhoods, not sets of edges.
%\ognjen{In table 4, and tables later I find $\subseteq 2^{\Triples \cup \Triples^{-}}$ confusing a bit, since it may give an impression that is imposing a restriction but actually it holds the definition... maybe its just me, so feel free to ignore my comment}

In what follows, for a positive integer $n$, we write $\te^n$ for $ \te \shexeach \ldots \shexeach \te$ where  $\te$ is repeated $n$ times,
$\te^{\leq n}$ for $\varepsilon \shexone \te^{1} \shexone \ldots \shexone \te^{n}$, and $\te^{\geq n}$ for $\te^{n} \shexeach \te^{*}$.
For a closed triple expression $\tte$, we let $\shexneighopen{\tte} = \shexneigh{\tte \shexeach (\shexneginv{R})^{*} \shexeach (\shexneg{Q})^{*}}$ where $Q$ is the set of predicates and keys that appear \emph{directly} in $\tte$ (as opposed to appearing in $\varphi$ for a sub-expression $q.\se$ of $\tte$) and $R$ is the set of predicates and keys whose inversions appear directly in $\tte$. For instance, if $\tte = p.\shexneigh{q.\shexneigh{\top}\shexeach\shexinverse{p}.\shexneigh{\top}}$, then $Q = \{p\}$ and $R=\emptyset$.

% \begin{itemize}
% % \item $\te^n$ for some positive integer $n$ denotes the triple expression $ \te \shexeach \ldots \shexeach \te$ where  $\te$ is used $n$ times,
% % \item $\te^{\leq n}$ as a short-hand for $\varepsilon \shexone \te^{1} \shexone \ldots \shexone \te^{n}$ and $\te^{\geq n}$ as a short-hand for $\te^{n} \shexeach \te^{*}$,
% \item
% \end{itemize}
\begin{example}
Let us now see how the concrete constraints from Example~\ref{ex:constraint-desc} can be handled in ShEx.
\begin{align*}
& \shexneigh{
\shexinverse{\excard}.\shextop \shexeach \shexallte} \Rightarrow   \test(\mathbbm{int})  & \mbox{(C1)} \\
& \shexneigh{\exowns.\shextop \shexeach \shexallte} \Rightarrow \shexneigh{ \exemail.\shextop%^{\geq 1}
 \shexeach \shexallte} & \mbox{(C2)} \\
&    \shexneigh{\shexinverse{\exemail}.\shextop \shexeach \shexallte} \Rightarrow \shexneighopen{ (\exemail^{-}.\shextop)^{\leq 1}   } & \mbox{(C3)} \\
& \shexneigh{\excard . \shextop \shexeach \shexallte}  \Rightarrow   \shexneighopen{\exprivileged.\neg \shexhasvalue(\mathit{true})} \lor  &  \\
& \qquad \shexneighopen{(\exaccess^{-}.\shexneighopen{\exprivileged.\shexhasvalue(\mathit{true})})^{*} }  & \mbox{(C4)} \\
&  \shexneigh{\exemail.\shextop \shexeach \shexallte}  \Rightarrow \shexneighopen{(\exaccess.\shextop)^{\leq 5}} & \mbox{(C5)}
\end{align*}
\end{example}

We next show a more complex example, which illustrates the power of ShEx that is not readily available in SHACL or PG-Schema.

\begin{example} \label{ex:SheXCounting}
    Suppose that we want to express the following constraint on each user who owns an account: the number of accounts to which the user has access is greater or equal to the number of accounts that the user owns. We can do this in ShEx as follows:
    \begin{flalign*}
    & \shexneigh{\exowns.\shextop \shexeach \shexallte} \Rightarrow  \\ & \shexneighopen{ (\exaccess.\shextop )^{*}\shexeach (\exowns.\shextop \shexeach \exaccess.\shextop)^{*}}
    \end{flalign*}

    Similarly to the above (yet more abstractly) consider the following requirement:  for the node $c$, the number of outgoing $p$-edges  is \emph{equal} to the number of outgoing $q$-edges. %This can be expressed via a ShEx schema $\shexschema^{eq}=\{\shexhasvalue(c)\Rightarrow \shexneighopen{  (p.\shextop \shexeach q.\shextop)^{*}}\}$.In Appendix \ref{app:indistinguishabilitySHACL} we prove that $\shexschema^{eq}$ cannot be expressed in SHACL.
This can be expressed in ShEx using $\shexhasvalue(c)\Rightarrow \shexneighopen{  (p.\shextop \shexeach q.\shextop)^{*}}$ but  cannot be expressed in SHACL (see Appendix \ref{app:indistinguishabilitySHACL})
   \end{example}

Finally, let us see why ShEx and SHACL count differently.

\begin{example}%[ShEx counts edges]
\label{ex:shex-counts-edges}
    The following SHACL schema ensures that from every node with an outgoing $\exaccess$-edge, exactly two nodes are accessible via a $\exaccess$-edge or an $\exowns$-edge:
    \[
    \exists \exaccess \Rightarrow \exists^{= 2} (\exaccess \cup
    \exowns).\top
    \]
    Here $\exists^{=n} \pexpr.\varphi$ is a shorthand for $\exists^{\leq n} \pexpr.\varphi \land \exists^{\geq n} \pexpr.\varphi$.
    For instance, in Figure~\ref{fig:example-shex-counts-edges}, the graph on the right is valid, whereas the one on the left is not.
    The same constraint cannot be expressed in ShEx because ShEx cannot distinguish these two graphs (see Appendix~\ref{app:indistinguishabilityShEx}).
    The reason is that ShEx triple expressions count triples adjacent to a node, whereas SHACL and PG-Schema count nodes on the opposite end of such triples.
    This makes counting edges simpler in ShEx: the ShEx shape $\shexneigh{(p. \shextop \shexone q.\shextop)^{2} \shexeach (\shexneginv{\emptyset})^{*}}$ allows exactly two outgoing edges labelled $p$ or $q$. In SHACL this is written as $(\exists^{=2} p. \top \land \exists^{=0} q.\top) \lor (\exists^{=2} q.\top \land \exists^{=0} p.\top) \lor (\exists^{=1} p.\top \wedge \exists^{=1} q.\top)$.
    %%% THIS PROPOSAL IS INCORRECT, BECAUSE IT OPENS OUTGOING EDGES.
    %\textcolor{red}{As a consequence, counting edges is simpler in ShEx: the triple expression $(p. \shextop \shexone q.\shextop)^{2} \shexeach \shexallte$ says that a node should have exactly two outgoing edges labelled $p$ or $q$. The same constraint in SHACL gives $(\exists^{=2} p. \top \land \exists^{=0} q.\top) \lor (\exists^{=2} q.\top \land \exists^{=0} p.\top) \lor (\exists^{=1} p.\top \wedge \exists^{=1} q.\top)$}.
\end{example}

\begin{figure}[t]
\resizebox{.9\linewidth}{!}{
  \includegraphics{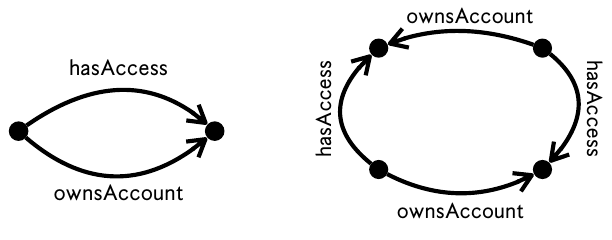}
}
\Description{A diagram showing two graphs indistinguishable by ShEx}
\caption{Two graphs indistinguishable by ShEx}
\label{fig:example-shex-counts-edges}
\end{figure}

%%% Local Variables:
%%% mode: latex
%%% TeX-master: "Article"
%%% End:

\section{Shape-based PG-Schema}
\label{sec:pgschema-simplified}

Shape-based PG-Schema is a non-recursive combination of a logic and two generative formalisms. It uses path expressions to specify paths (as in SHACL), and \emph{content types} to specify node contents. Both path expressions and content types are then used in formulas defining shapes.
%% \todo[inline]{Fabio: At the beginning of SHACL section we say that SHACL is also based on a combination of a logic and a generative formalism. I think we need to make this distinction clearer or to remove it tout court.}
%% \todo[inline]{Filip: Good point! Rewrote this. Hope it's fine now.}
Content types in PG-Schema play a role similar to triple expressions in ShEx, but they are only used for properties. Because all properties of a node must have different keys, they are much simpler than triple expressions (in fact, they can be translated into a fragment of SHACL). Unlike for SHACL and ShEx, the abstraction of shape-based PG-Schema departs significantly from the original design. Original PG-Schema uses queries written in an external query language,  which is left unspecified aside from some basic assumptions about the expressive power. Here we use a specific query language (PG-path expressions). Importantly, up to the choice of the query language, the abstraction we present here faithfully captures the expressive power of the original PG-Schema. A detailed comparison can be found in Appendix~\ref{sec:standard-pg-schema}.

\begin{definition}[Content type]
\label{def:contentType}
A \emph{content type} is an expression $\ntype$ of the form defined by the grammar
 $$\ntype \gDef \top \ \gMid \ \closedRT{}\  \gMid\  \closedRT{k : \vtype} \ \gMid\   \ntype \tAnd \ntype  \ \gMid \  \ntype \tOr \ntype  \gEnd$$
where $k \in \Keys$ and $\vtype \in \ValueTypes$. 
\end{definition}

Recall that $\Records$ is the set of all records (finite-domain partial functions $r : \Keys \pto \Values$). We write $\emptyRec$ for the empty record. 
% For records $r_1$ and $r_2$, we let $r_1 \cup r_2$ be the function that behaves as $r_1$ on $\dom{r_1}$ and as $r_2$ on $\dom{r_2}$. 
% We require that $r_1(k) = r_2(k)$ for every $k \in \dom{r_1} \cap \dom{r_2}$. 
The semantics of content types is defined in Table~\ref{tab:semPG-content-types}. Note that $\sem{\ntype}$ is independent from $\graph$ and can be infinite.

\begin{table}
  \caption{Semantics of content types.}
  \label{tab:semPG-content-types}  
  \centering
  \begin{tabular}{cl}
    \toprule
    $\ntype$ & $\sem{\ntype} \subseteq \Records$ \\
    \midrule    
     $\sem{\top}$ & $\Records$ \\[2pt]
     $\sem{\closedRT{}}$ & $\{ \emptyRec \}$\\[2pt]
     $\sem{\closedRT{k : \vtype}}$ & $\big\{ \{(k, w)\} \ \big|\ w \in \sem{\vtype} \big\}$\\[2pt]
    $\sem{\,\ntype_1 \tAnd \ntype_2\,}$ & $\{ (r_1 \cup r_2) \in \Records \mid r_1 \in \sem{\ntype_1} \wedge r_2  \in \sem{\ntype_2} \}$\\[2pt]
    $\sem{\,\ntype_1 \tOr \ntype_2\,}$ & $\sem{\ntype_1} \cup \sem{\ntype_2}$\\[2pt]
    \bottomrule
  \end{tabular}
\end{table}

\begin{example}
We assume integers and strings are represented via $\mathbbm{int},\mathbbm{str}\in \ValueTypes$.  Suppose we want to create a content type for nodes that have a string value for the $\exemail$ key and \emph{optionally} have an integer value for the $\excard$ key. No other key-value pairs are allowed. We should then use $\closedRT{\exemail : \mathbbm{str}} \tAnd (\closedRT{\excard : \mathbbm{int}} \tOr \closedRT{})  $.
\end{example}

%\todo[inline]{ Give more insight about the choice of the query language.
%based on a variant of path expressions, that is natural from the perspective of property graphs and
%} 

\begin{definition}[PG-path expressions] 
\label{def:pgpaths-syntax}
A PG-path expression is an expression $\pexpr$ of the form defined by the  grammar
\begin{align*} 
& \pexpr \gDef  \ppexpr \gMid \ppexpr \cdot k \gMid k^{-} \cdot \ppexpr \gMid k^{-} \cdot \ppexpr \cdot  k' \gEnd \\
& \ppexpr \gDef \keyIsVal{k}{c} \gMid \neg \keyIsVal{k}{c} \gMid \ntype \gMid \lnot \ntype \gMid p \gMid \lnot P \gMid  
{\ppexpr}^{-} \gMid \ppexpr \cdot \ppexpr \gMid \ppexpr \cup \ppexpr \gMid {\ppexpr}^{*} \gEnd
\end{align*}
where $k,k' \in \Keys$, $c \in \Values$, $\ntype$ is a content type, $p \in \Predicates$, and $P \subseteq_{\mathit{fin}} \Predicates$. We use $k$, $k^{-}$, and $k^{-}\cdot k'$ as short-hands for PG-path expressions $\top\cdot k$, $k^{-}\cdot \top$, and $k^{-}\cdot \top\cdot  k'$, respectively. 
\end{definition}

% \todo[inline]{JH: I would prefer notation $[k = c]$ over $\keyIsVal{k}{c}$. That makes it more clear we are filtering and not navigating away by following edges or keys, and it avoids better the impression that this is similar to a closed type. TODO: Jan does this.}

Unlike in SHACL, PG-path expressions cannot navigate freely through values. In the property graph world, this would correspond to a join, which is a costly operation. Indeed, existing query languages for property graphs do not allow joins under ${}^{*}$. However, PG-path expressions can start in a value and finish in a value. This leads to \emph{node-to-node}, \emph{node-to-value}, \emph{value-to-node}, and \emph{value-to-value} PG-path expressions, reflected in the four cases in the first rule of the grammar.

The semantics of PG-path expression $\pexpr$ for graph $\graph$ is a binary relation over
$\nodes(\graph) \cup \values(\graph)$, defined in Table~\ref{tab:semPGtypes}. In the table, $k$ is treated as any other subexpressions, eventhough it can only be used at the end of a PG-path expression, or in the beginning as $k^{-}$.
Notice that $\lnot \ntype$ matches nodes whose content is not of type $\ntype$,  
$\lnot P$ matches edges with a label that is not in $P$ (in particular, $\lnot\emptyset$ matches all edges). 
Also, 
$\gsem{\pexpr}$ is always a subset of $\Nodes \times \Nodes$, $\Nodes \times \Values$, $\Values \times \Nodes$, or $\Values \times \Values$, corresponding to the four kinds of PG-path expressions discussed above.

%\todo[inline]{SA: Shouldn't it be $\gsem{\pexpr}\subseteq \nodes(\graph) \cup \values(\graph)$ in Table 4?}
%\todo[inline]{Filip: I guess you mean: \[\gsem{\pexpr}\subseteq (\nodes(\graph) \cup \values(\graph))\times (\nodes(\graph) \cup \values(\graph))\]
%It would be more precise, but it's really long... I am not even sure it is necessary at all in the table. It was meant as a recollection of the "type" of the returned objects, at Fabio's request. }

\begin{table}[tb]
  \caption{Semantics of PG-path expressions.}
  \label{tab:semPGtypes}  
  \centering
  \begin{tabular}{cl}
    \toprule
    $\pexpr$ & $\gsem{\pexpr}\subseteq (\Nodes\cup\Values)\times (\Nodes\cup\Values)\ $  for  $\graph = (E, \rho)$ \\[2pt]
    \midrule    
    $\keyIsVal{k}{c}$ & $\left\{ (u, u) \mid u \in \nodes(\graph) \land (k,c) \in \rho(u) \right\}$ \\[2pt]
        $\neg \keyIsVal{k}{c}$ & $\left\{ (u, u) \mid u \in \nodes(\graph) \land (k,c) \notin \rho(u) \right\}$ \\[2pt]
    $\ntype$ & $\left\{ (u, u) \mid u\in\nodes(\graph)\land \rho(u) \in \sem{\ntype} \right\}$ \\[2pt]
    $\lnot \ntype$ & $\left\{ (u, u) \mid u \in \nodes(\graph)\land \rho(u)\notin \sem{\ntype} \right\}$ \\[2pt]
    $k$ & $\{(u,w)\mid \rho(u,k)=w\}$ \\[2pt]
    $p$ & $\left\{ (u, v) \mid (u, p, v) \in E \right\}$  \\[2pt]
    $\lnot P$ & $\left\{ (u, v) \mid \exists p : (u, p, v) \in E \wedge p \notin P \right\}$ \\[2pt]
    $\pathExpr^{-}$ & $\left\{(u,v)\mid (v,u) \in \iexpr{\pathExpr}{\graph}\right\}$ \\[2pt]
    $\pathExpr \cdot \pathExpr'$ & $\left\{(u,v) \mid \exists w: (u,w)\in\iexpr{\pathExpr}{\graph}\land (w,v)\in\iexpr{\pathExpr'}{\graph}\right\}$ \\[2pt]
    $\pathExpr\cup \pathExpr'$ & $\iexpr{\pathExpr}{\graph}\cup\iexpr{\pathExpr'}{\graph}$\\[2pt]
   $\pathExpr^{*}$ & $ \{(u,u)\mid u\in\nodes(\graph)\}\cup \iexpr{\pathExpr}{\graph}  \cup \iexpr{\pathExpr \cdot \pathExpr}{\graph} \cup \ldots $ \\
    \bottomrule
  \end{tabular}
\end{table}

\begin{definition}[PG-Shapes]
A PG-Shape is an expression $\varphi$ defined by the following grammar:
\[ 
\varphi \gDef \exists ^{\leq n} \, \pexpr \gMid \exists^{\geq n} \,\pexpr \gMid \varphi \land \varphi \gEnd
\]
where $\pexpr$ is a PG-path expression. We use $\exists$ and $\nexists$ as short-hands for $\exists^{\geq 1}$ and $\exists^{\leq 0}$.
\end{definition}

The semantics of PG-shapes  is defined in Table~\ref{tab:semPGshapes}. We say $v\in\Nodes\cup\Values$ \emph{satisfies} a PG-shape $\varphi$ in a graph $\graph$ if $\graph, v \models \varphi$. Every PG-shape is satisfied by nodes only or by values only. 

\begin{table}[tb]
  \caption{Satisfaction of PG-shapes}
  \label{tab:semPGshapes}  
  \centering
  \begin{tabular}{cl}
    \toprule
    $\varphi$ & $\graph, v \models\varphi\ $ for  $v\in\Nodes\cup\Values$\\
    \midrule   
    $\exists^{\leq n} \, \pexpr$ &
    $\#\left\{v' \mid (v, v') \in \gsem{\pexpr}\right\} \leq n$  \\[2pt]
    $\exists^{\geq n} \, \pexpr$ &
    $\#\left\{v' \mid (v, v') \in \gsem{\pexpr}\right\} \geq n$  \\[1pt]
    $\varphi_1 \land \varphi_2$ &
    $\graph, v \models\varphi_1$ and $\graph, v \models\varphi_2$ \\
    \bottomrule
  \end{tabular}
\end{table}

%\subsection{PG-Schemas}

\begin{definition}[PG-Selectors]
A \emph{PG-selector} is a PG-shape of the form $\exists\, \pexpr$. 
\end{definition}

A \emph{PG-Schema} $\schema$ is 
a finite set of pairs $(\textit{sel}, \varphi)$ where $\textit{sel}$ is a PG-selector and $\varphi$ is a PG-shape.
The semantics of PG-Schemas is defned just like in Section ~\ref{ssec:shapes}.

\begin{example}\label{ex:sharedExamplesPGS} The constraints (C1-C5) from Example~\ref{ex:constraint-desc} can be handled in PG-Schema as follows: 
\begin{align*}
& \exists \excard \Rightarrow  \exists \big ( \closedRT{\excard : \mathbbm{int}}\tAnd \top \big )  & \mbox{(C1)} \\
&  \exists \exowns\Rightarrow \exists \exemail  & \mbox{(C2)} \\
&  \exists\exemail^{-}  \Rightarrow \exists^{\leq 1} \exemail^{-} & \mbox{(C3)} \\
& \exists\, (\{ \excard : \any  \} \tAnd \top) \cdot  \{\exprivileged :\mathit{true} \} \Rightarrow  & \\ &   \qquad  
\nexists \,\exaccess^{-} \cdot \neg \{\exprivileged :\mathit{true}\} & \mbox{(C4)} \\
% & \exists\exowns \Rightarrow \exists \exemail  & \\
& \exists \exemail \Rightarrow \exists ^{\leq 5} \, \exaccess & \mbox{(C5)} 
\end{align*}
Notice that in rule (C1), we indeed need $\exists \excard,$ rather than $\exists \excard^{-}$, because there is no PG-Shape to state that the selected value is of type $\mathbbm{int}$, and so we formulate C1 as a statement about nodes.
%\todo[inline]{JH: this is just to flag that the previous remark was added. It was originally suggested as a footnote, but I feel it is important enough to be in the text. Originally the formulation was that PG-Selectors cannot select values (which is not really true) or that PG-Shapes cannot make statements about values (which is also not really true).}
% \todo[inline]{JH: the solution for C1 should have selector $\exists \excard$, and not $\exists \excard^{-}$ as it does now. TODO: Fixed, but maybe add a foot note about this. @Jan}
\end{example}
% \todo[inline]{Cem: $\not \exists \pi$ not directly expressible in PG-Shapes, so C4 should be rewritten}
% \todo[inline]{Mantas: hm, but there is $\nexists$ defined as a shorthand for $\exists^{\leq 0}$. Am I missing something here?}

% \end{example}

A characteristic feature of PG-Schema, revealing its database provenience, is that it can close the whole graph by imposing restrictions on all nodes.

\begin{example}
\label{ex:closedgraph}
    Given a common graph such as the one in~\Cref{fig:common-graph}, we might want to express that each node has a key $\exprivileged$ with a boolean value and  either a key $\excard$ with an integer value or a key $\exemail$ with a string value, and no other keys are allowed. In PG-Schema this  can be expressed as follows:
    \[ \exists \top \Rightarrow \exists \{\exprivileged: \mathbbm{bool}\} \tAnd\big(\{ \excard : \mathbbm{int} \} \tOr \{ \exemail : \mathbbm{str} \}\big) \,.\]
    We can also forbid  any predicates except those mentioned in the running example:
    \[ \exists \top \Rightarrow \not\exists \lnot \{\exowns,\exaccess,\exinvited\} \,.\]
\end{example}

\section{Common Graph Schema Language}
\label{sec:core}

We now present the Common Graph Schema Language (CoGSL), which combines the core functionalities shared by SHACL, ShEx, and PG-Schema (over common graphs). 

Let us begin by examining the restrictions that need to be imposed. We shall refer to shapes and selectors used in CoGSL as \emph{common shapes} and \emph{common selectors}. 
Common shapes cannot be closed under disjunction and negation, because PG-Schema shapes are purely conjunctive. For the same reason common shapes cannot be nested. 
Kleene star ${}^{*}$ cannot be allowed in path expressions because we consider ShEx without recursion. 
% \todo{Drop the following sentence?}
% By switching to ShEx with recursion, we would be able to support arbitrary SHACL path expressions in shapes of the form $\exists \pi$, but not arbitrary PG-path expressions as these are too expressive for SHACL. 
%
%\todo[inline]{also reading this sentence, the last part is not clear whether should be SHACL or ShEx. Do we also want to say something about not allowing nesting of shapes in the common shapes? It is mentioned in Section 5, but wonder whether it should be mentioned here as well. } 
%%E.g., can one express in PG-Schema the SHACL constaint $\exists p \Rightarrow \leq_2 p_1.=_3 p_2$?}
%%% FILIP: Rephrased. Added a sentence about nesting in the previous paragraph.  
%
Supporting path expressions traversing more than one edge under counting quantifiers is impossible as this is not expressible in ShEx. Supporting disjunctions of labels of the form $p_1 \cup p_2$ is also impossible, due to a mismatch in the approach to counting: while SHACL and PG-Schema count nodes and values, ShEx counts triples, as illustrated in \Cref{ex:shex-counts-edges}. 

Closed content types and $\lnot P$ cannot be used freely, because neither SHACL nor ShEx are capable of closing only properties or only predicate edges: both must be closed at the same time. 

Finally, selectors are restricted because SHACL and ShEx do not support $\top$ as a selector; that is, one cannot say that each node (or value) in the graph satisfies a given shape. This means that SHACL and ShEx schemas always allow a disconnected part of the graph that uses only predicates and keys not mentioned in the schema, whereas PG-Schema can disallow it (see Example~\ref{ex:closedgraph}). 

Putting these restrictions together we obtain the Common Graph Schema Language. We define it below as a fragment of PG-Schema. 

\begin{definition}[common shape]
\label{def:simple-shape}
    A \emph{common shape} $\varphi$ is an expression given by the grammar
\begin{align*}
\varphi  \gDef  & 
 \exists\,\pexpr
\gMid \exists^{\leq n} \, \pexpr_1
\gMid \exists^{\geq n} \, \pexpr_1 \gMid 
\exists\, \ntype \land \not\exists\, \lnot P
\gMid \varphi \land \varphi \gEnd\\
\ntype \gDef &\closedRT{}\  \gMid\  \closedRT{k : \vtype} \ \gMid\   \ntype \tAnd \ntype  \ \gMid \  \ntype \tOr \ntype  \gEnd\\
\pexpr_0 \gDef & \keyIsVal{k}{c} \gMid 
\lnot \keyIsVal{k}{c} \gMid \ntype\tAnd\top\gMid \lnot (\ntype\tAnd\top) \gMid \pexpr_0 \cdot \pexpr_0 \gEnd\\
\pexpr_1 \gDef &  \pexpr_0  \cdot p \cdot
\pexpr_0 
\gMid  \pexpr_0  \cdot p^{-} \cdot
\pexpr_0  \gMid \pexpr_0\cdot k \gMid k^{-}\cdot \pexpr_0 \gEnd\\
\ppexpr \gDef & \pexpr_0  \gMid p 
\gMid {\ppexpr}^{-} \gMid \ppexpr \cdot \ppexpr 
\gMid \ppexpr \cup \ppexpr \gEnd \\
\pexpr \gDef & \ppexpr \gMid \ppexpr \cdot k \gMid k^{-}\cdot \ppexpr \gMid k^{-}\cdot \ppexpr\cdot k' \gEnd
\end{align*}
where $n \in \mathbb{N}$, $P\subseteq_{\mathit{fin}} \Predicates$, $k,k'\in\Keys$, $c\in\Values$, and $p\in\Predicates$. 
\end{definition}

That is, $\ntype$ is a content type that does not use $\top$ (a \emph{closed} content type),  $\pi_0$ is a PG-path expression that always stays in the same node (a \emph{filter}), $\pi_1$ is a PG-path expression that traverses a single edge or property (forward or backwards), and $\pi$ is a PG-path expression that uses neither ${}^{*}$ nor $\lnot P$. 
Moreover, $\pi_0$, $\pi_1$, and $\pi$ can only use \emph{open} content types; that is, content types of the form $\ntype \tAnd \top$.
The use of $\lnot P$ is limited to closing the neighbourhood of a node (this is the only way PG-Schema can do it).

% \begin{definition}[Simple selector]
% \label{def:simple-selector}
% A \emph{simple selector} is a common shape of one of the following forms 
% \[ 
% \exists\, p \,, \quad 
% \exists\, p^{-} \,,\quad 
% \exists\, k \,,\quad
% \exists\, \{k:w\} \,, \quad \exists\, k^{-}\,,
% \]
% where $k\in\Keys$,  $p\in\Predicates$, and $w\in \Values$. 
% \end{definition}

\begin{definition}[common selector]
\label{def:common-selector}
A \emph{common selector} is a common shape of one of the following forms
\[ 
\exists\, k \,,\;
\exists\, p \cdot \pexpr\,, \;
\exists\, p^{-}\!\cdot \pexpr \,,\;
\exists\, \keyIsVal{k}{c}\cdot \pexpr \,, \;
\exists\, \big(\{k:\vtype\}\tAnd\top\big)\cdot \pexpr \,, \;
\exists\, k^{-}\!\cdot \pexpr\,,\!\!
\]
where $k\in\Keys$,  $p\in\Predicates$, $c\in \Values$, $\vtype\in\ValueTypes$ and $\pexpr = \ppexpr$ or $\pexpr = \ppexpr\cdot k'$ for some PG-path expression $\ppexpr$ generated by the grammar in Definition~\ref{def:simple-shape} and some $k'\in\Keys$.
\end{definition}

That is, a common selector is a common shape of the form $\exists\, \pexpr$ such that  the PG-path expression $\pexpr$ requires the focus node or value to occur in a triple with a specified predicate or key.

A \emph{common schema} is a finite set of  pairs $(\textit{sel}, \varphi)$ where $\textit{sel}$ is a common selector and $\varphi$ is a common shape. The semantics is inherited from PG-Schema. 

%We illustrate common schemas in \Cref{ex:sharedExamplesPGS}, which describes PG-Schemas, but all of which are also valid common schemas.
We note that we showed that the constraints (C1)-(C5) from our running example can be expressed in all three formalisms. Specifically, the PG-Schema representation from \Cref{ex:sharedExamplesPGS} is also a common schema. 

\begin{proposition}
\label{prop:core}
For every common schema there exist equivalent SHACL and ShEx schemas. 
\end{proposition}

The translation is relatively straightforward (see Appendix~\ref{sec:appendix-core}). The two main observations are that star-free PG-path expressions can be simulated by nested SHACL and ShEx shapes, and that closure of SHACL and ShEx shapes under Boolean connectives  allows encoding complex selectors in the shape (as the antecedent of an implication). We illustrate the latter in~\Cref{ex:PathSelector}.

\begin{example}[Complex paths in selectors] \label{ex:PathSelector}
We want to express that all users who have invited a user who has invited someone (so there is a path following two $\exinvited$ edges) must have a key $\exemail$ of type $\Exvt{str}$.
In PG-schema we express this as:
\[ \exists \exinvited \cdot \exinvited \Rightarrow \{ \exemail : \Exvt{str}  \} \tAnd \top     \]

At first glance, it seems unclear how to express this in the other formalisms, since they do not permit paths in the selector. However, we can see that paths in selectors can be encoded into the shape: 

\noindent In SHACL, using the same example, we do this by 
\begin{align*}
\exists \exinvited  \Rightarrow &\  \neg (\exists \exinvited \cdot \exinvited ) \lor  \exists \exemail . \test(\Exvt{str})      
\end{align*}
And in  ShEx for this example would be: 
\begin{align*}
\shexneigh{\exinvited.\shextop \shexeach \shexallte}  \Rightarrow &\ 
\neg \se_2
\lor \shexneighopen{\exemail.\test(\Exvt{str})}  
\end{align*}
where $\se_2 = \shexneighopen{\;
    \left(
    \exinvited. \se_1
    \right)^{\geq 1}
\;}$ and $\se_1 = \shexneighopen{\,
        \exinvited . \shextop^{\geq 1}
    \,}$.
That is, $\se_1$ is satisfied by nodes that have an outgoing path $\exinvited$, and $\se_2$ by nodes that have an outgoing path $\exinvited\cdot\exinvited$.
For paths of unbounded length, it is not apparent how such a translation would proceed for ShEx schemas in the absence of recursion. 
\end{example}

\section{Related Work}

\paragraph{SHACL literature.} The authoritative source for SHACL is the W3C recommendation~\cite{KK17}.
%that first introduced SHACL
Further literature on SHACL following its standardisation can be roughly divided into two groups. The first group studies the formal properties and expressiveness of the non-recursive fragment~\cite{BJVdB24}. Notable examples in this category are: the work by Delva et al. on data provenance~\cite{DDJB23}, the work of Pareti et al. on satisfiability and (shape) containment~\cite{PKMN20}, and the work of Leinberger et al. connecting the containment problem to description logics~\cite{LSRLS20}.
The second group of papers is concerned with proposing a suitable semantics for recursive SHACL~\cite{CRS18,CFRS19,ACORSS20,BJ21} or studying the complexity of certain problems for recursive SHACL under a chosen semantics~\cite{PKM22}. First reports on practical applications and use-cases for SHACL include the study of expressivity of property constraints, as well as  mining and extracting constraints in the context of large knowledge graphs such as Wikidata and DBpedia~\cite{FSAP24,RLH23}. Finally, the underlying ideas of SHACL where transposed to the setting of Property Graphs in a formalism called ProGS \cite{ProGS}.

\paragraph{ShEx literature}

ShEx was initially proposed in 2014 as a concise and human-readable language to describe, validate, and transform RDF data~\cite{PGS14}.
Its formal semantics was formally defined in~\cite{SBG15}.
The semantics of ShEx schemas combining recursion and negation was later presented in~\cite{BGP17}.
The current semantic specification of the ShEx language has been published as a W3C Community group report~\cite{PBGK19} and a new language version is currently being defined as part of the IEEE Working group on Shape Expressions\footnote{\url{https://shex.io/shex-next/}}. As for practical applications, ShEx has been applied as a descriptive schema language through the
Wikidata Schemas project\footnote{\url{https://www.wikidata.org/wiki/Wikidata:WikiProject_Schemas}}.
%Additional work went into extending ShEx to handle graph models that go beyond RDF, like WShEx to validate Wikibase graphs~\cite{L22}, ShEx-Star to handle RDF-Star and PShEx to handle property graphs~\cite{L24}; while these wokrs extend ShEx to different types of property graphs and does not look for a common graph model, so it has a very different focusneither of ths and does not look for a common graph model, so it has a very different focus
%Wikidata Schemas project\footnote{\url{https://www.wikidata.org/wiki/Wikidata:WikiProject_Schemas}}.
Additional work went into extending ShEx to handle graph models that go beyond
RDF, like WShEx to validate Wikibase graphs~\cite{L22}, ShEx-Star to handle
RDF-Star and PShEx to handle property graphs~\cite{L24}. While these works
extend ShEx to (different types of) property graphs, they do not provide a
common graph data model \todo{nor compare schema languages, as we do} that
allows comparing schema languages, as we do.

\paragraph{PG-Schema literature.}

PG-Schema, as introduced in~\cite{ABDF23}, builds upon an earlier proposal of PG-Keys~\cite{ABDF21} to enhance schema support for property graphs, in the light of limited schema support in existing systems and the current version of the GQL standard~\cite{GQL}. It is currently being used in the GQL standardization process as a basis for a standard for property graph schemas.

\paragraph{Comparing RDF schema formalisms.}

In Chapter 7 of \cite{GPBK17}, the authors compare common features and differences between ShEx and SHACL and~\cite{GGFE19} presents a simplified language called S, which captures the essence of ShEx and SHACL. Tomaszuk~\cite{T17} analyzes advances in RDF validation, highlighting key requirements for validation languages and comparing the strengths and weaknesses of various approaches.
% The following sentence is about comparing schemas for RDF and property graphs

\paragraph{Interoperability between schema graph formalisms.}

Interoperability between schema graph formalisms like RDF and Property Graphs remains challenging due to differences in structure and semantics. RDF focuses on triple-based modeling with formal semantics, while Property Graphs allow flexible annotation of relationships with properties. RDF-star \cite{H14} and RDF 1.2 \cite{KCHS24} extend RDF 1.1 by enabling statements about triples, aligning more closely with LPG: for instance, RDF-star allows triples to function as subjects or objects, similar to how LPG edges carry properties.

By adopting \emph{named graphs}~\cite{CARROLL2005247}, already RDF 1.1 provided a mechanism for making statements about (sub-)graphs. Likewise, different \emph{reification} mechanisms have been proposed in the literature for RDF in order to ``embed'' meta-statements about triples (and graphs) in ``vanilla'' RDF graphs, ranging from the relatively verbose original W3C reification vocabulary as part of the original RDF specification, to more subtle approaches such as singleton property reification~\cite{NBS14}, which is close to the unique identifiers used for edges in most LPG models. Custom reification models are used, for instance, in Wikidata, to map Wikibase's property graph schema to RDF, cf. e.g.~\cite{FSAP24,HHK15}.
There is also work on  schema-independent and schema-dependent methods for transforming RDF into Property Graphs, providing formal foundations for preserving information and semantics \cite{ATT20}.
All these approaches, in principle, facilitate general or specific mappings between RDF and LPGs, which is what the present paper tries to avoid by focusing on a common submodel.

There have been several prior proposals for uniying graph data models, rather then providing mappings between them.  The OneGraph initiative \cite{LSHBBB23} aims to bridge the different graph data models by promoting a unified graph data model for seamless interaction. Similarly, MilleniumDB's Domain Graph model~\cite{DCRMD23} aims at covering RDF, RDF-star, and property graphs. These works seek a common \emph{supermodel}, aiming to support a both RDF and LPGs via more general solutions. In contrast, we aim at understanding the existing schema languages by studying them over a common submodel of RDF and LPGs.

\paragraph{Schemas for tree-structured data.}

The principle of defining (parts of) schemas as a set of pairs $(sel,\varphi)$ is also used in schema languages for XML. A DTD~\cite{xml} is essentially such a set of pairs in which $sel$ selects nodes with a certain label, and $\varphi$ describes the structure of their children. In XML Schema, the principle was used for defining key constraints (using \emph{selectors} and \emph{fields}) \cite[Section~3.11.1]{xsd}. The equally expressive language BonXai~\cite{MNNS17} is based on writing the entire schema using such rules. Schematron \cite{schematron} is another XML schema language that differs from grammar-based languages by defining patterns of assertions using XPath expressions~\cite{xpath}. It excels in specifying constraints across different branches of a document tree, where traditional schema paradigms often fall short. Schematron's rule-based structure, composed of phases, patterns, rules, and assertions, allows for the validation of documents.

\paragraph{RDF validation}

Last, but not least, it should be noted that the requirement for (constraining) schema languages—besides ontology languages such as OWL and RDF Schema—in the Semantic Web community is much older than the more recent additions of SHACL and ShEx. Earlier proposals in a similar direction include efforts to add constraint readings of Description Logic axioms to OWL, such as OWL Flight \cite{BRP05} or OWL IC \cite{S10}. Another approach is Resource Shapes (ReSh) \cite{R14}, a vocabulary for specifying RDF shapes. The authors of ReSh recognize that RDF terms originate from various vocabularies, and the ReSh shape defines the integrity constraints that RDF graphs are required to satisfy. Similarly, Description Set Profiles (DSP) \cite{N08} and SPARQL Inferencing Notation (SPIN) \cite{KHI11} are notable alternatives. While SHACL, ShEx, and ReSh share declarative, high-level descriptions of RDF graph content, DSP and SPIN offer additional mechanisms for validating and constraining RDF data, each with its own strengths and applications.

\paragraph{Implementations}
Dozens of tools support graph data validation, including ShEx and SHACL. A comprehensive collaborative list of resources is available at:
\url{https://github.com/w3c-cg/awesome-semantic-shapes}.

\section{Conclusions}

We provided a formal and comprehensive comparison of the three most prominent schema languages in the Semantic Web and Graph Database communities: SHACL, ShEx, and PG-Schema. Through painstaking discussions within our working group, we managed to (1) agree on a common data model that captures features of both RDF and Property Graphs and (2) extract, for each of the languages, a core that we mutually agree on, which we define formally. Moreover, the definitions of (the cores of) each of the schema languages on a common formal framework allows readers to maximally leverage their understanding of one schema language in order to understand the others. Furthermore, this common framework allowed us to extract the Common Graph Schema Language, which is a cleanly defined set of functionalities shared by SHACL, ShEx, and PG-Schema. This commonality can serve as a basis for future efforts in integrating or translating between the languages, promoting interoperability in applications that rely on heterogeneous data models. For example, we want to investigate recursive ShEx and more expressive query languages for PG-Schema more deeply.

\begin{section}*{Acknowledgments}

This work was initiated during Dagstuhl Seminar 24102 \emph{Shapes in Graph
Data}.
It was funded by the Austrian Science Fund (FWF) [10.55776/COE12] (Polleres);
ANR project EQUUS ANR-19-CE48-0019, project no.~431183758 by the German Research Foundation (Martens);
ANGLIRU: Applying kNowledge Graphs to research data interoperabiLIty and
ReUsability, code: PID2020-117912RB from the Spanish Research Agency (Labra
Gayo);
European Union's Horizon Europe research and innovation program under Grant
Agreement No 101136244 (TARGET) (Hose and Tomaszuk);
Austrian Science Fund (FWF) and netidee SCIENCE [T1349-N], and the Vienna
Science and Technology Fund (WWTF) [10.47379/ICT2201] (Ahmetaj);
Poland's NCN grant 2018/30/E/ST6/00042 (Murlak);
and FWF stand-alone project P30873 (\v{S}imkus).
F. Mogavero is member of the Gruppo Nazionale Calcolo Scientifico-Istituto
Nazionale di Alta Matematica.

\end{section}

% End of file Acknowledgments.tex

\bibliographystyle{ACM-Reference-Format}

\bibliography{references,ReferencesFM}

\normalsize

\appendix

\section{Distilling the common data model}
\label{sec:appendix-foundations}

\todo[inline]{Remember to add what we promised to the reviewers. @Filip}

In this section we discuss the relationship between common graphs and the standard data models of the three schema formalisms formalisms---RDF and property graphs.

\subsection{Comparison with RDF}
\label{app:sec-foundations-comparison-rdf}

As explained in Section~\ref{sec:prl}, common graphs can be naturally seen as finite sets of triples from 
$\Triples = \left(\Nodes\times\Predicates\times\Nodes\right) \;\cup\; \left(\Nodes\times\Keys\times\Values\right)$, with  $(E,\rho)$ corresponding to 
$E \;\cup\; \{(u,k,v) \mid \rho(u,k) = v\}$. 

Unlike in RDF, a common graph may contain at most one tuple of the form $(u,k,v)$ for each $u\in \Nodes$ and $k\in\Keys$. This reflects the assumption that  properties are single-valued, which is present in the property graph data model. 

In the RDF context, one would assume the following:
\begin{itemize}
\item $\Nodes \subseteq \IRIs \cup \Blanks$,
\item $\Predicates \subseteq \IRIs$, 
\item $\Keys \subseteq \IRIs$,
\item $\Values = \Literals$.
\end{itemize}
However, the common graph data model does not refer to $\IRIs$, $\Blanks$, and $\Literals$ at all, because these are not part of the property graph data model. 

In contrast to the RDF model, but in accordance with the perspective commonly taken in databases, both values and nodes are atomic. For nodes we completely abstract away from the actual representation of their identities. We do not even distinguish between $\IRIs$ and $\Blanks$. An immediate consequence of this is that schemas do not have access to any information about the node other than the triples in which it participates. In particular, they cannot compare nodes with constants. This is a significant restriction with respect to the RDF data model, but it follows immediately from the same assumption made in the property graph data model. On the positive side, this aspect is entirely orthogonal to the main discussion in this paper, so eliminating it from the common data model does not oversimplify the picture. 

For values we take a more subtle approach: we assume a set $\ValueTypes$ of value types, with each $\vtype\in\ValueTypes$ representing a set $\sem{\vtype}\subseteq\Values$. This captures uniformly data types, such as \texttt{integer} or \texttt{string}, and user-defined checks, such as interval bounds for numeric values or regular expressions for strings.
On the other hand, the common graph data model does not include any binary relations over values, such as an order.  

\todo[inline]{Iovka: The section Comparison with Proprety graphs includes a hint on how a common graph could encode a general property graph. Similar encoding could be considered also for RDF: keys \texttt{iri} and \texttt{blank} for node identities, nodes with special key \texttt{value} to represent literal values (which would also lift the constraint of not being able to give to different literal values for the same predicate). We haven't considered such encoding approach in the present paper, but it seems to be a valid question that might be worth investigating.}

\todo[inline]{After discussion, we decided to write a small paragraph about the possibility to encode RDF in common graphs, in the spirit of the similar paragraph for property graphs.}
%\todo[inline]{Maxime: Discussing these ``encoding'' or ``translation'' approaches might be an interesting Appendix onto itself. 

%One issue I currently see is that while property graphs might be ``easily'' encoded into the common data model, RDF graphs might not be. For example, how would we handle multiple values for a given key?

%Another subtle point is that our common graph datamodel treats values are a ``navigational'' object, a little bit like nodes. This is important for RDF/SHACL to express for example the key constraint. You could argue that this is more of a ``schema thing'' than it is a ``data model thing''. I am unsure. It might come down to the idea of a ``node'' in a graph. In RDF speak, values are nodes, while in LPGs they are not.}

%We note that in common graphs, there are no simple mechanisms to encode classes of objects and that this is a difference to the three formalisms that we compare, which do expose mechanisms to easily expressing classes in some form (in the form of node labels in Property Graphs, and classes in RDF). We believe our decision to omit this is justified because ... 

\subsection{Comparison with property graphs}
\label{sect:PGCGComparison}

Let us recall the standard definition of property graphs \cite{ABDF23}. 

\begin{definition}[Property graph]
A \emph{property graph} is a tuple $(N, E, \pi, \lambda, \rho)$ such that 
\begin{itemize}
\item $N$ is a finite set of nodes;
\item $E$ is a finite set of edges, disjoint from $N$; 
\item $\pi : E \to (N \times N)$ maps edges to their source and target;
\item $\lambda : (N \cup E) \to 2^{\Predicates}$ maps nodes and edges to finite sets of labels;
\item $\rho : (N \cup E)\times \Keys \pto \Values$ is a finite-domain partial function mapping element-key pairs to values.
\end{itemize}
\end{definition}

A common graph $G = (E', \rho')$ can be easily represented as a property graph by letting
\begin{itemize}
    \item $N = \nodes(G)$,
    \item $E = E'$,
    \item $\pi = \{ (e, (v_1, v_2)) \mid e = (v_1, p, v_2) \in E \}$,
    \item $\lambda = \{ (e, \{ p \}) \mid e = (v_1, p, v_2) \in E \} \cup \{ (v, \emptyset) \mid v \in N \}$, and
    \item $\rho = \rho'$.
\end{itemize}

It is possible to characterise exactly the property graphs that are such representations of common graphs. These are the property graphs $(N, E, \pi, \lambda, \rho)$ for which it holds that:
\begin{enumerate}
    \item $\lambda(v) = \emptyset$ for all $v \in N$, and $\lambda(e)$ is a singleton for all $e \in E$,
    \item there cannot be two distinct edges $e_1, e_2 \in E$ such that $\pi(e_1) = \pi(e_2)$ and $\lambda(e_1) = \lambda(e_2)$, and
    \item $\rho(e, k)$ is undefined for all $e \in E$, $k \in \Keys$.
\end{enumerate}

So, common graphs can be interpreted as restricted property graphs: no labels on nodes, single labels on edges, no parallel edges with the same label, and no properties on edges. All these restrictions are direct consequences of the nature of the RDF data model. 

While these restrictions seem severe at a first glance, the resulting data model can actually easily simulate unrestricted property graphs: labels on nodes can be simulated with the presence of corresponding keys, edges can be materialised as nodes if we need properties over edges or parallel edges with the same label. This means not only that common graphs can be used without loss of generality in expressiveness and complexity studies, but also that the corresponding restricted property graphs are flexible enough to be usable in practice, while additionally guaranteeing interoperability with the RDF data model. 

\subsection{Class information}
\todo[inline]{Dominik: it might be helpful to include examples or more detailed explanations on how to use designated predicates and keys to simulate class memberships and hierarchies. Providing these examples would clarify how common graphs can indirectly support class information despite not having direct mechanisms for it.}
The common graph data model does not have direct support for class information. The reason for this is that RDF and property graphs handle class information rather differently. In RDF, both class and instance information is part of the graph data itself:  classes are elements of the graph, subclass-superclass relationships are represented as edges between classes, and membership relationships are represented as edges between elements and classes. In property graphs, the membership of a node in a class is indicated by a label put on the node. A node can belong to many classes, but the only way to say that class $A$ is a subclass of class $B$ is to ensure in the schema that each node with label $A$ also has label $B$. That is, 
\begin{itemize}
\item in property graphs class membership information is available locally in a node, but consistency must be ensured by the schema, 
\item in RDF, obtaining class membership information requires navigating in the graph, but consistency is for free.
\end{itemize}
Clearly, both approaches have their merits, but when passing from one to the other data needs to be translated. This means that we cannot pick one of these approaches for the common data model while keeping it a natural submodel of both RDF and property graphs. Therefore, to reduce the complexity of this study, we do not include any dedicated features for supporting class information in our common data model. 
Note, however, that common graphs can support both these approaches indirectly: designated predicates can be used to represent membership and subclass relationships, and keys with a dummy value can simulate node labels.

%!TEX root =  Article.tex

\section{Standard SHACL}

\label{app:standard-shacl}

Standard SHACL defines shapes as a conjunction of \emph{constraint components}. The different constructs from our formalization correspond to fundamental building blocks of these constraint components. Next to that, the formalization of SHACL presented in this paper is less expressive than standard SHACL. First, because we define it here for the common data model (which corresponds to a strict subset of RDF, see \Cref{app:sec-foundations-comparison-rdf}), and second, because we want to simplify our narrative: we leave out the comparison of RDF terms using \texttt{sh:lessThan} for this reason. Furthermore, because the common data model omits language tags, the corresponding constraint components from standard SHACL are omitted as well.

Our formalization is closely tied to the ones found in the literature. There, correspondence between the formalization of the literature and standard SHACL has been shown in detail \cite{MJPHD}. This section highlights and discusses some relevant details.

\paragraph{Class targets and constraint component.} As a consequence of the common data model not directly supporting the modelling of classes, some class-based features are not adapted in our formalization. Specifically, there are no selectors (``target declarations'' in standard SHACL) that involve classes. Furthermore, the value type constraint component \texttt{sh:class} is not covered by our formalization.

%The common data model does not support defining instances of a certain class. Even though both property graphs and RDF have support for this, they both tackle it very differently: in RDF, class and instance information is part of the graph data itself, while in property graphs it is more of a modelling aspect. A concrete issue in this regard is that RDF allows for navigating the graph in order to determine class information, while this is not possible in property graphs where nodes are labelled directly. This makes it difficult to compare the two approaches, as a common data model is then no longer a subset of RDF and property graphs --- it would require a form of translation. Therefore, to reduce the complexity of our approach, we decided not to encode class information in our common data model. A consequence of this is that we leave out all class related features of SHACL. Specifically, the class constraint component, and the ability to target nodes of a certain SHACL class. 

%Technically, you could encode class information, as it exists in RDF, in this datamodel because it is part of the data. However, this is then difficult to reconcile with PG-Schema, because ...

%\todo[inline]{Filip: I copied the paragraph above to the first appendix. Here it will probably be sufficient to repeat briefly what disappears from the schema language. }

\paragraph{Closedness.} In standard SHACL syntax, closedness is a property that takes a true or false value. The semantics of closedness is based on a list of predicates that are allowed for a given focus node. This list can be inferred based on the predicates used in property shapes, or this list can be explicitly given using the \texttt{sh:ignoredProperties} keyword. Our formalization effectively adopts the latter approach: $\closed(Q)$ means that the properties mentioned in the set $Q$ are the ignored properties.

\paragraph{Path expressions.} The path expressions used in our formalization deviate from the standard in three obvious ways. First, we make a distinction between `keys' and `predicates'. This is simply a consequence of using our common data model. Second, we leave out some of the immediately available path constructs from standard SHACL: one-or-more paths and zero-or-one path. However, these are expressible using the building blocks of our formalization: one-or-more paths are expressed as $\pi\cdot\pi^{*}$, and zero-or-one paths are expressed as $\pi\cup\id$. Lastly, our path expression allow for writing the identity relation explicitly. This cannot be done in literal standard SHACL syntax, but its addition to the formalization serves to highlight its hidden presence in the language. Writing the identity relation directly in a counting construct, e.g., $\exists^{\geq n}\id.\top$, never adds expressive power. In the case of $n=1$, the shape is always satisfied (and thus equivalent to $\top$), and it is easy to see that for any $n > 1$, it is never satisfied. The situation with complex path expressions in counting constructs is less clear from the outset. However, it has been shown \cite{BJVdB24} (Lemma 3.3), that the only case where $\id$ adds expressiveness is with complex path expressions of the form $\pi\cup\id$. This is exactly the definition of zero-or-one paths and is therefore covered by standard SHACL. Another place where $\id$ can occur in our formalization is in the equality and disjointness constraints, e.g., $\eq(\id,p)$. According to the standard SHACL recommendation, you cannot write this shape. However, in the SHACL Test Suite \cite{LKK24} test \texttt{core/node/equals-001}, the following shape is tested for:
\begin{verbatim}
ex:TestShape
  rdf:type sh:NodeShape ;
  sh:equals ex:property ;
  sh:targetNode ex:ValidResource1 .
\end{verbatim}
on the following data:
\begin{verbatim}
ex:ValidResource1
  ex:property ex:ValidResource1 .
\end{verbatim}
The intended meaning of this test is, in natural language: ``The targetnode \texttt{ex:ValidResource1} has a \texttt{ex:property} self-loop and no other \texttt{ex:property} properties''. Effectively, this is the semantics for our $\eq(\id,p)$ construct. The situation with $\disj(\id,p)$ is similar.

We therefore have an ambiguous situation: the standard description of SHACL does not allow for shapes of the form $\eq(\id,p)$, but the test suite, and therefore all implementations that pass it completely, do\footnote{Incidentally, all implementations currently mentioned in the implementation report handle these cases correctly.}. It then seems fair to include this powerful construct in the formalization. 

\paragraph{Comparisons with constants.}
A direct consequence of the assumption that node identities in the common data model are hidden from the user, our abstraction of SHACL on common graphs does not support comparisons with constants from $\Nodes$. Comparisons with constants from $\Values$ are allowed.

\paragraph{Node tests.} 

Our formalization uses the $\test(\vtype)$ construct to denote many of the node tests available in SHACL. We list the tests from standard SHACL that are covered by this construct.

\smallskip
\noindent
$\bullet$
{\bf \texttt{DatatypeConstraintComponent}}  \\ Tests whether a node has a certain datatype. \\
\noindent
$\bullet$ {\bf\texttt{MinExclusiveConstraintComponent}} or \\ \noindent \phantom{$\bullet$} {\bf\texttt{MinInclusiveConstraintComponent}} or \\ \noindent \phantom{$\bullet$} {\bf\texttt{MaxExclusiveConstraintComponent} }or \\ \noindent \phantom{$\bullet$} {\bf\texttt{MaxInclusiveConstraintComponent}} \\
These four constraints cover can check whether a node is larger (\texttt{\textbf{Max}}) or smaller (\texttt{\textbf{Min}}) than some value, and whether this forms a partial order (\texttt{\textbf{Inclusive}}) or a strict, or total, order (\texttt{\textbf{Exclusive}}).
Based on the SPARQL $<$ or $\leq$ operator mapping. \\
\noindent
$\bullet$ {\bf\texttt{MaxLengthConstraintComponent}} or \\ \noindent \phantom{$\bullet$}   {\bf\texttt{MinLengthConstraintComponent} } \\
  These two constraints test whether the
  length of the lexical form of the node is ``larger'' or equal (resp. ``smaller'' or equal) than
  some provided integer value. Strictly speaking, the recommendation defines these constraint components also on IRIs. However, we limit their use to Literals. \\
\noindent
$\bullet$ {\bf\texttt{PatternConstraintComponent}} \\ Tests whether the length of the lexical form of the node satisfies some regular expression. Strictly speaking, the recommendation defines these constraint components also on IRIs. However, we limit their use to Literals.

\smallskip 

Then there are two tests not covered by our formalization:

\smallskip

\noindent
$\bullet$ {\bf\texttt{NodeKindConstraintComponent}} \\
Tests whether a node is an IRI, Blank Node, or Literal. Our tests apply only to RDF Literals.\\
\noindent
$\bullet$ {\bf\texttt{LanguageInConstraintComponent}} \\
 Test whether the language tag of the node is one of the specified language tags. This feature is not supported by our data model, since it lacks language tags.

\section{Standard ShEx}
\label{app:shex-appendix}

\newcommand{\stshex}{s-ShEx\xspace}

The Shape Expressions Language (ShEx)~\cite{PBGK19} and the ShapeMaps
language~\cite{PB19} have been defined by the Shape Expressions Community
group\footnote{\url{https://www.w3.org/community/shex/}} at W3C.
%
% (already in related work)
%The semantics of standard ShEx language has been formalised first in
%\cite{SBG15} for triple expressions, then in \cite{BGP17} for the full
%language, but without inverse predicates.
%
Hereafter, we use \emph{standard ShEx} or \emph{\stshex} to refer to the
language defined in~\cite{PBGK19} and formalised in~\cite{SBG15,BGP17}, while
\emph{ShEx} designates the language presented in the current work.

In this appendix, we support the following
\begin{claim}
\label{claim:app-shex}
  On common graphs, the expressive power of ShEx schemas is equivalent to the
  expressive power of non-recursive \stshex schemas.
\end{claim}

Section~\ref{app:sec-shex-standard-shex-on-common-graphs} explains \stshex on
common graphs, while Section~\ref{app:sec-shex-recursion} explains non-recursive
\stshex.

\subsection{\stshex schema and the validation problem}

A standard ShEx schema is a set of named shape expressions, and it is usually
formalised as a pair $(L, \mathit{def})$, where $L$ is a finite set of shape
names (in practice, these are IRIs) and $\mathit{def}$ is a function that
associates a shape expression with every shape name.
In \stshex, the validation problem $\graph \models \schema$ from
Section~\ref{ssec:shapes} is defined in a different way.
In fact, the ShEx specification~\cite{PBGK19} does not specify what it means for
a graph to be valid \wrt a \stshex schema; it only defines what it means for a
node in a graph to satisfy a shape expression.

However, the problem considered in practice is whether some selected nodes in
the graph satisfy some prescribed shape expressions from the schema.
This is specified by a shape map~\cite{PB19}.
A shape map can be formalised as a set of pairs of the form $(\mathit{sel},
l)$, where $l \in L$ and $\mathit{sel}$ is a unary query.
While the shape map specification~\cite{PB19} allows the selectors from
Definition~\ref{def:shex-selector}, most implementations allow general SPARQL
queries as selectors.

In the current paper, we integrate the shape map into the schema itself, which
allows us to specify the validation problem in a uniform way for the three graph
schema formalisms considered.
Additionally, in shape maps we do not use shape names, but shape expressions
directly; the next section argues why this is not a problem from the point of
view of expressive power for standard ShEx without recursion.

\subsection{Shape names and recursion}
\label{app:sec-shex-recursion}

Recursion is an important mechanism in standard ShEx.
% In this work, however, we consider only ShEx schemas without recursion (to be
% defined shortly).
% This restriction has been made because neither standard SHACL\footnote{There
% are SHACL formalisations that introduce recursion, but the semantics of this
% construct is not specified in the standard.} nor PG-Schema allow for recursion.
% In standard ShEx, 
The fact that shape expressions are named permits to refer to
them using their name.
In particular, these references allow for circular recursive definitions.
As an example, consider the standard ShEx schema in
Figure~\ref{fig:app-shex-schema-example}.
It contains the single shape name \texttt{ex:User}, whose definition is given by
the shape expression inside the curly braces.
The latter shape expression refers to itself: \texttt{@ex:User} indicates a
reference to the shape expression named \texttt{ex:User}.
Concretely, the shape expression requires from an RDF node to have an
\texttt{ex:email} predicate whose value is a string, as well as any number of
\texttt{ex:invited} predicates whose values are nodes that satisfy the shape
expression named \texttt{ex:User}.

\begin{figure}[h]
\centering
{\small
\begin{verbatim}
  PREFIX ex: <http://ex.example/#>
  PREFIX xsd: <http://www.w3.org/2001/XMLSchema#>
  ex:User {
    ex:email   xsd:string ;
    ex:invited @ex:User *
  }
\end{verbatim}}
\vspace{-2em}
\caption{\label{fig:app-shex-schema-example}%
  A standard ShEx schema.}
\end{figure}

A \stshex schema $(L, \mathit{def})$ is called recursive, when there is a shape
name $l \in L$ whose definition $\mathit{def}(l)$ uses a reference \texttt{@}$l$
to itself, either directly or transitively through references to other shape
names.
Every standard non-recursive ShEx schema can be rewritten to an equivalent
schema without references, simply by replacing every reference with its
definition.
In other words, references in standard ShEx do not add expressive power for
non-recursive schemas.
Therefore, here after we will present standard ShEx without references.

\subsection{Syntax of standard ShEx on common graphs}
\label{app:sec-shex-standard-shex-on-common-graphs}

We present here a version of \stshex restricted on common graphs.
The principle difference between \stshex described here and the ShEx
recommendation~\cite{PBGK19} resides in so called \emph{node
constraints}\footnote{\url{http://shex.io/shex-semantics/\#node- constraints}}.
These are constraints to be verified on the actual node of an RDF graph (which
is an IRI, a literal or a blank node) without considering its neighbourhood.
As pointed out in Section~\ref{app:sec-foundations-comparison-rdf}, such
constraints are irrelevant for nodes (\ie, elements of $\Nodes$) in common
graphs.
Therefore, we restrict \stshex node constraints on values (elements of
$\Values$) only.
\stshex node constraints on values correspond to the atomic shape expressions
$\shextest(\vtype)$ and $\shexhasvalue(c)$ in ShEx.
Their expressive power can be entirely captured by selecting for $\ValueTypes$ a
language equivalent to node constraints on values in \stshex.
Note finally that $\shextest(\any)$ in ShEx allows to distinguish nodes from
values.

\newcommand{\stse}{\mathit{se}}
\newcommand{\stte}{\mathit{te}}
\newcommand{\stclosed}{\mathsf{closed}}
\newcommand{\stextra}{\mathsf{extra}}
\newcommand{\stshape}{\mathit{sh}}
\newcommand{\sttc}{\mathit{tc}}
\newcommand{\stcard}{\mathit{card}}
\newcommand{\normalised}[1]{\tilde{#1}}
\newcommand{\trtefromst}[1]{\tau_{\text{e}}(#1)}
\newcommand{\trfromst}[1]{\tau(#1)}
\newcommand{\trtetost}[1]{\sigma_{\text{e}}(#1)}
\newcommand{\trtost}[1]{\sigma(#1)}
\newcommand{\stses}[2]{\mathit{SE}_{#2}(#1)}
\newcommand{\stextraconstr}[2]{\eta^{#2}_{#1}}
\newcommand{\preds}{\mathit{preds}}

Figure~\ref{fig:app-standard-shex-syntax} gives an abstract syntax for \stshex
\emph{shape expressions} $\stse$ and \emph{triple expressions} $\stte$
restricted on common graphs \wrt node constraints, as discussed above.
The non-terminal $\mathit{sh}$ corresponds to \emph{Shape}s, while the
non-terminal $\mathit{tc}$ is for
\emph{TripleExpression}s\footnote{\url{http://shex.io/shex-semantics/\#shapes-
and-TEs}}.
This abstract syntax is intended to be understandable by those familiar with
standard ShEx after taking into account the following purely syntactic
differences:
\begin{itemize}
\item
  we use the mathematical notation $\land$, $\lor$ and $\neg$ for the
  \stshex operators $\mathsf{and}$, $\mathsf{or}$ and $\mathsf{not}$;
\item
  according to the ShEx specification, the $\stextra\ Q$ modifier is optional
  for shapes; however, an absent extra set is equivalent to $\stextra \
  \emptyset$, therefore we will consider that it is always present.
\end{itemize}

\begin{figure}[h]
\begin{align*}
  \stse
\gDef
  &     \shexhasvalue(c)
  \gMid \shextest(\text{$\vtype$})
  \gMid \stshape
  \gMid \stse \land \stse
  \gMid \stse \lor \stse
  \gMid \neg \stse
\gEnd \\
  \stshape
\gDef
  &     \stextra\ Q\ \shexneigh{\stte}
  \gMid \stclosed\ \stextra\ Q\ \shexneigh{\stte}
\gEnd \\
  \stte
\gDef
  &     \sttc
  \gMid \stte \shexeach \stte
  \gMid \stte \!\shexone\! \stte
  \gMid \stte {[\mathit{min}; \mathit{max}]}
\gEnd \\
  \sttc
\gDef
  &     q\ \stse
  \gMid q\ .
\gEnd
\end{align*}
with $c \in \Nodes \cup \Values$, $\vtype \in \ValueTypes$, $q \in \Predicates
\cup \Keys \cup \Predicates^{-} \cup \Keys^{-}$, $Q \subseteq_{\mathit{fin}}
\Predicates \cup \Keys \cup \Predicates^{-} \cup \Keys^{-}$, $\mathit{min} \in
\mathbb{N}$ and $\mathit{max} \in \mathbb{N} \cup \{*\}$.
\caption{\label{fig:app-standard-shex-syntax}%
  Abstract syntax for \stshex.}
\end{figure}

\subsection{Translations between ShEx and \stshex}

In this section, we introduce a back and forth translation between non-recursive
\stshex and ShEx.
We claim that these translations preserve the semantics \wrt the validity of a
graph.
The claim is presented without a correctness proof, as it
would require to introduce here a formal semantics for \stshex.
However, it is not difficult to write a proof making use of the formal semantics
from~\cite{BGP17}.

\subsubsection{Differences between \stshex and ShEx}

We now list the syntactic differences between the two languages, and describe
how they are handled by the translation:
\begin{itemize}[\textbullet]
\item
  Triple constraints in \stshex (non-terminal $\sttc$) allow to use a $.$
  (dot) instead of the shape expression, which is in fact equivalent to the
  ShEx shape expression $\shexneigh{\shexallte}$.
\item
  Triple expressions in ShEx contain the atomic expression $\varepsilon$,
  while \stshex does not allow it directly.
  On the other hand \stshex allows us to use intervals of the form
  $[\mathit{min};\mathit{max}]$ to define bounded or unbounded repetition,
  while ShEx allows only the unbounded repetition $*$.
  We show in Section~\ref{app:sec-shex-syntax-triple-expr} that the two
  variants have equivalent expressive power.
\item
  In \stshex, the atomic shape expression that defines the neighbourhood of a
  node (non-terminal $\stshape$) is parametrised by a set $Q$ of \emph{extra}
  (possibly inverse) predicates and keys.
  In Section~\ref{app:sec-shex-eliminate-extra}, we show that $\stextra$ is
  just syntactic sugar in \stshex.
\item
  In \stshex, the atomic shape expression derived from the non-terminal
  $\stshape$ can have an optional $\stclosed$ modifier.
  On the other hand, ShEx introduces the triple expressions $\shexneg{P}$ and
  $\shexneginv{P}$ (for $P \subseteq \Predicates \cup \Keys$).
  As we will see, the latter are used when translating \stshex to ShEx in
  order to distinguish between $\stclosed$ and non-$\stclosed$ \stshex
  shape expressions.
\end{itemize}

\subsubsection{Normalised triple expressions}
\label{app:sec-shex-syntax-triple-expr}

We now show how both \stshex and ShEx triple expressions can be normalised so as
to use a limited number of operators; this shall be useful for the translation
between \stshex and ShEx.

\paragraph{Normalisation of \stshex triple expressions}

A \stshex triple expression is called \emph{normalised} if it uses only the
intervals $[0; 1]$ and $[0; *]$; these can be normalised using rewriting rules
based on the following equivalences:
\begin{align*}
  \stte[\mathit{min};*]
& =
  \stte[0;*] \shexeach \underbrace{\stte \shexeach \cdots \shexeach
  \stte}_{\mathit{min} \text{ times}}
\\
  \stte[\mathit{min};\mathit{max}]
& =
  \underbrace{\stte \shexeach \cdots \shexeach \stte}_{\mathit{min} \text{
  times}} \shexeach \underbrace{\stte[0;1] \shexeach \cdots \shexeach
  \stte[0;1]}_{\mathit{max}-\mathit{min} \text{ times}} \quad \text{ when }
  \mathit{max} \not= *
\end{align*}

\paragraph{Normalisation of ShEx triple expressions}

Here after, $\tte$ designates a \todo{cosed!}closed ShEx triple expression derivable from the
rule $f$ of the grammar in Definition~\ref{def:shex-syntax}.
For every $\tte$, we define $\tte^{?} = \tte \shexone \varepsilon$.
A triple expression $\tte$ is \emph{normalised} if either $\tte = \varepsilon$,
or $\tte$ does not use $\varepsilon$ as sub-expression, but can use the $?$
operator defined above.
Every triple expression can be normalised by eliminating occurrences of
$\varepsilon$ using the $?$ operator and the following two properties:
\begin{itemize}
\item
  $\varepsilon$ is a neutral element for the $\shexeach$ operator, \ie, $\tte
  \shexeach \varepsilon = \varepsilon \shexeach \tte = \tte$ for every ShEx
  triple expression $\tte$,
\item
  $\varepsilon^{*} = \varepsilon$.
\end{itemize}
\Wlogx, \textbf{from now on we only consider normalised triple expressions}.

\subsubsection{Direct predicates of triple expressions}
\label{app:sec-shex-define-preds-of-triple-expr}

This section is devoted to the introduction of two technical definitions.
For every triple expression we define the set of (possibly inverted) predicates
and keys that appear directly in the expression.
Formally, if $\tte$ is a ShEx triple expression derived by the \todo{check with grammar} third rule of
the grammar in Definition~\ref{def:shex-syntax}, then we define the set
$\preds(\tte) \subseteq \Predicates \cup \Keys \cup \Predicates^{-} \cup
\Keys^{-}$ inductively on the structure of $\tte$ by:
\begin{align*}
  \preds(\varepsilon)
& =
  \emptyset
\\
  \preds(p.\se)
& =
  \{ p \}
\\
  \preds(\shexinverse{p}.\se)
& =
  \{ \shexinverse{p}\}
\\
  \preds(\se \shexeach \se')
& =
  \preds(\se) \cup \preds(\se')
\\
  \preds(\se \shexone \se')
& =
  \preds(\se) \cup \preds(\se')
\\
  \preds(\se^{*})
& =
  \preds(\se)
\end{align*}
For a \stshex triple expression $\stte$, the set $\preds(\stte)$ is defined
similarly (recall that $q \in \Predicates \cup \Keys \cup \Predicates^{-} \cup
\Keys^{-}$):
\begin{align*}
  \preds(q\ \stse)
& =
  \{ q \}
\\
  \preds(q\ .)
& =
  \{ q \}
\\
  \preds(\stse \shexeach \stse')
& =
  \preds(\stse) \cup \preds(\stse')
\\
  \preds(\stse \shexone \stse')
& =
  \preds(\stse) \cup \preds(\stse')
\\
  \preds(\stte[\mathit{min}; \mathit{max}])
& =
  \preds(\stse)
\end{align*}

\subsubsection{Eliminating $\stextra$ from \stshex}
\label{app:sec-shex-eliminate-extra}

We show by means of an example how the extra construct can be eliminated in
\stshex.
%The same example will be used later on for the translation from standard ShEx
%to ShEx, therefore it is described in detail.

\begin{example}
\label{ex:app-shex-extra}
Consider the \stshex shape expression
\begin{align*}
  \stse
& =
  \stextra\ \{p_1, p_2\}\ \ \shexneigh{\stte}
\\
  \text{with} \qquad \stte
& =
  p_1 \shexneigh{p\ .} \,\shexeach\, p_1 \shexneigh{p'\ .} \,\shexeach\, p_3\ .
\\ %\,\shexeach\, \shexinverse{p}_4\ . \\
  \text{and}\qquad \quad
&
  \quad p_1, p_2, p_3, p, p' \in \Predicates \cup \Keys
\end{align*}
that has a set of extra predicates and keys $\{ p_1, p_2 \}$.
It is satisfied by nodes whose neighbourhood can have any incoming triples and
has the following outgoing triples:
\begin{enumerate}
\item
  one $p_1$-triple leading to a node that satisfies $\shexneigh{p\ .}$,
\item
  another $p_1$-triple leading to a node that satisfies $\shexneigh{p'\ .}$,
\item
  a $p_3$-triple leading to an unconstrained node,
\item
  \label{req:extrap1}
  because $p_1$ appears in the $\stextra$ set, other $p_1$-triples are also
  allowed as long as they satisfy \textbf{none} of the constraints present for
  $p_1$ in $\stte$, that is, they satisfy neither $\shexneigh{p\ .}$ nor
  $\shexneigh{p'\ .}$,
\item
  \label{req:extrap2}
  because $p_2$ appears in the $\stextra$ set too, $p_2$-triples are allowed and
  their target is not constrained because $p_2$ does not appear in the triple
  expression $\stte$,
\item
  \label{req:closed}
  finally, since the shape is not $\stclosed$, all outgoing triples whose
  predicate is not in $\{ p_1, p_3 \}$ are allowed, noting that $\{ p_1, p_3 \} =
  \preds(\stte) \cap \Predicates \cap \Keys$.
\end{enumerate}
% The node has also the following incoming triples:
% \begin{enumerate}
% \addtocounter{enumi}{6}
% \item one incoming $p_4$-triple coming from an unconstrained node,
% \item because standard ShEx does not close the incoming triples, all incoming triples whose predicate is different from $p_4$ are allowed.
% \end{enumerate}
\end{example}
The shape expression $\stse$ from Example~\ref{ex:app-shex-extra} is equivalent
to the following shape expression without extra:
\[
  \shexneigh{ \stte \;\shexeach\; \stte_{p_1}^{*} \;\shexeach\; \stte_{p_2}^{*}}
\]
where
\[
  \stte_{p_1}
=
  p_1\ \left( \neg \shexneigh{p\ .} \land \neg \shexneigh{p'\ .} \right)
\qquad \text{ and } \qquad
  \stte_{p_2}
=
  p_2\ .
\]
The sub-expression $\stte_{p_1}^{*}$ allows to satisfy the
requirement~(\ref{req:extrap1}) from Example~\ref{ex:app-shex-extra}, while the
sub-expression $\stte_{p_2}^{*}$ allows to satisfy the
requirement~(\ref{req:extrap2}).

The construction for eliminating $\stextra$ just discussed can be generalised
to arbitrary shape expressions.
The idea is to combine (with the $\shexeach$ operator) the initial triple
expression with a sub-expression of the form $q\ \stse_q^{*}$ for every
(possibly inverse) $\stextra$ predicate $q$, where $\stse_q$ is the conjunction
of the negated shape expressions $\stse'$ such that $q\ \stse'$ appears directly
in $\stte$ (without traversing any shape expressions $\stse$).

\Wlogx, \textbf{from now on, we consider only \stshex shape expressions without
$\stextra$}.

\subsubsection{Translation from \stshex to ShEx}

With every \stshex shape expression $\stse$ we associate the ShEx shape
expression $\trfromst{\stse}$ as defined in
Table~\ref{tab:app-shex-translation-from-standard}.
It is defined by mutual recursion with the corresponding translation function
$\trtefromst{\stte}$ for standard ShEx triple expressions $\stte$ presented in
Table~\ref{tab:app-shex-translation-from-standard-te}.

\begin{table}[ht]
\caption{\label{tab:app-shex-translation-from-standard-te}%
  Translation from \stshex to ShEx for normalised triple expressions $\stte$,
  with $q \in \Predicates \cup \Keys \cup \Predicates^{-} \cup \Keys^{-}$.}
\centering
\begin{tabular}{ll}
\toprule
  $\stte$                   & $\trtefromst{\stte}$ \\
\midrule
  $q\ \stse$                & $p. \,\trfromst{\stse}$ \\
  $q\ .$                    & $p.\shextop$ \\
  $\stte \shexeach \stte'$  & $\trtefromst{\stte} \shexeach
                              \trtefromst{\stte'}$ \\
  $\stte \shexone \stte'$   & $\trtefromst{\stte} \shexone
                              \trtefromst{\stte'}$ \\
  $\stte[0;*]$              & $\trtefromst{\stte}^{*}$ \\
  $\stte[0;1]$              & $\trtefromst{\stte} \shexone \varepsilon$ \\
\bottomrule
\end{tabular}
\end{table}

\begin{table}[ht]
\caption{\label{tab:app-shex-translation-from-standard}%
  Translation from \stshex to ShEx for shape expressions.}
\centering
\begin{tabular}{ll}
\toprule
  $\stse$                         & $\trfromst{\stse}$ \\
\midrule
  $\shexhasvalue(c)$              & $\shexhasvalue(c)$ \\
  $\trfromst{\shextest(\vtype)}$  & $\shextest(\vtype)$ \\
  $\stse \land \stse'$            & $\trfromst{\stse} \land
                                    \trfromst{\stse'}$ \\
  $\stse \lor \stse'$             & $\trfromst{\stse} \lor \trfromst{\stse'}$\\
  $\neg \stse$                    & $\neg \trfromst{\stse}$ \\
  $\stclosed\ \shexneigh{\stte}$  & $\shexneigh{\,\trtefromst{\stte}
                                    \,\shexeach\, ({\shexneginv{R}})^{*}\,}$ \\
                                  & \qquad with $R = \preds(\stte) \cap
                                    (\Predicates^{-} \cup \Keys^{-})$ \\
  $\shexneigh{\stte}$             & $\shexneigh{\,\trtefromst{\stte}
                                    \,\shexeach\, ({\shexneginv{R}})^{*}
                                    \,\shexeach\, (\shexneg{Q})^{*}\,}$ \\
                                  & \qquad with $R = \preds(\stte) \cap
                                    (\Predicates^{-} \cup \Keys^{-})$ \\
                                  & \qquad and $Q = \preds(\stte) \cap
                                    (\Predicates \cup \Keys)$ \\
\bottomrule
\end{tabular}
\end{table}

\OMIT{
% The example is too complex, not sure that this is a clarification
\begin{example}
The standard ShEx shape expression $\stse$ from Example~\ref{ex:app-shex-extra}
gives the following in ShEx:
$$
\shexneigh{\; \te \;\shexeach\; \te_{p_1}^{*} \;\shexeach\; \te_{p_2}^{*}
\;\shexeach\; (\shexneginv{Q})^{*} \;\shexeach\; (\shexneg{P})^{*}}
$$
where $\te = \trtefromst{\stte}$, $\te_{p_1}^{*}$ accounts for the extra $p_1$,
$\te_{p_2}^{*}$ accounts for the extra $p_2$, $\neg Q$ are the predicates open
for incoming edges, and $\neg P$ are the predicates open for outgoing edges.
\begin{align*}
\te = &\ p_1. \se_1 \shexeach p_1. \se'_1 \shexeach p_3.\shextop \shexeach
\shexinverse{p_4}.\shextop\\
\te_{p_1} = &\ \\
\te_{p_2} = &\ \\
Q = &\ \\
P = &\ \\
%   \se_1 = &\ \shexneigh{
%   p.\shextop \shexeach (\shexneginv{\emptyset}^{*} \shexeach (\shexneg{\{p\}})^{*}
%   }
\end{align*}
\begin{align*}
\stse &= \stextra\ \{p_1, p_2\}\ \ \shexneigh{\stte} \\
\text{with}\qquad \stte &= p_1 \shexneigh{p\ .} \,\shexeach\, p_1 \shexneigh{p'\ .} \,\shexeach\, p_3\ . \,\shexeach\, \shexinverse{p}_4\ . \\
\text{and}\qquad \quad & \quad p_1, p_2, p_3, p_4, p, p' \in \Predicates \cup \Keys
\end{align*}

$$
\shexneigh{ \stte \;\shexeach\; \stte_{p_1}^{*} \;\shexeach\; \stte_{p_2}^{*}}
$$
where
$$\stte_{p_1} = p_1\ \left(\neg \shexneigh{p\ .} \land \neg \shexneigh{p'\ .}\right)
\qquad \text{ and } \qquad\stte_{p_2} = (p_2\ .)$$

\end{example}
}%end OMIT

\subsubsection{Translation from ShEx to \stshex}

Unless otherwise specified, in the sequel, $\tte$ designates a closed \todo{closed!} ShEx
triple expression produced by the non-terminal $\tte$ of the grammar in
Definition~\ref{def:shex-syntax}.
In Table~\ref{tab:app-shex-translation-to-standard} we present a function that
with every normalised ShEx triple expression $\tte$ associates the standard ShEx
triple expression $\trtetost{\tte}$.
It is defined by mutual recursion with the translation function that with every
ShEx shape expression $\se$ associates a standard ShEx shape expression
$\trtost{\se}$, and that will be presented shortly.
Note that the case $\tte = \varepsilon$ is omitted in
Table~\ref{tab:app-shex-translation-to-standard}: recall that in normalised ShEx
triple expressions, $\varepsilon$ can only appear standalone (not in
sub-expressions), therefore the case $\tte = \varepsilon$ will be treated with
shape expressions.

\begin{table}[ht]
\caption{\label{tab:app-shex-translation-to-standard}%
  Translation from ShEx to \stshex for normalised triple expressions.}
\centering
\begin{tabular}{ll}
\toprule
  $\tte$                  & $\trtetost{\tte}$ \\
\midrule
  $p.\se$                 & $p\ \trtost{\se}$\\
  $\shexinverse{p}.\se$   & $\shexinverse{p}\ \trtost{\se}$\\
  $\tte \shexeach \tte'$  & $\trtetost{\tte} \shexeach \trtetost{\tte'}$\\
  $\tte \shexone \tte'$   & $\trtetost{\tte} \shexone \trtetost{\tte'}$\\
  $\tte^{*}$              & $\trtetost{\tte}[0;*]$\\
  $\tte^{?}$              & $\trtetost{\tte}[0;1]$\\
\bottomrule
\end{tabular}
\end{table}

The definition of $\trtost{\se}$ is straightforward for the following cases:
\begin{align*}
  \trtost{\shexhasvalue(c)}
& =
  \shexhasvalue(c) \\
  \trtost{\shextest(\vtype)}
& =
  \shextest(\vtype) \\
  \trtost{\se \land \se'}
& =
  \trtost{\se} \land \trtost{\se'} \\
  \trtost{\se \lor \se'}
& =
  \trtost{\se} \lor \trtost{\se'} \\
  \trtost{\neg \se}
& =
  \neg \trtost{\se}
\end{align*}
\todo{check with new terminology}
The remaining case is for a shape expression of the form $\shexneigh{\te} =
\shexneigh{\tte \shexeach \cdots}$ where $\tte$ is normalised.
Consider the most general case
\[
  \te
=
  \tte \shexeach (\shexneginv{P})^{*} \shexeach (\shexneg{Q})^{*}
\]
Let also
\begin{align*}
  \{ p_1, \ldots, p_m \}
& =
  (\preds(\tte) \cap \Predicates^{-} \cap \Keys^{-}) \setminus P \\
  \{ q_1, \ldots, q_n \}
& =
  (\preds(\tte) \cap \Predicates \cap \Keys) \setminus Q.
\end{align*}
Intuitively, $\{ p_1, \ldots, p_m \}$ is the set of predicates that are not
allowed to appear on incoming edges in the neighbourhoods defined by $\te$, and
similarly $\{q_1, \ldots, q_n\}$ are the forbidden predicates for outgoing
edges.
Then
\[
  \trtost{\shexneigh{\te}}
=
  \left\{
  \begin{array}{l}
    \trtetost{\tte} \;\shexeach \\
    \shexinverse{p_1}\ .[0;0] \shexeach \cdots \shexeach \shexinverse{p_m}\
    .[0;0] \;\shexeach \\
    q_1\ .[0;0] \shexeach \cdots q_n\ .[0;0]
  \end{array}
  \right\}
\]
If $\tte = \varepsilon$, then the term $\trtetost{\tte}$ on the first line of
the definition of $\trtost{\shexneigh{\te}}$ must be omitted.

The remaining case for the definition of $\trtost{\shexneigh{\te}}$ is for
\[
\te = \tte \shexeach (\shexneginv{P})^{*}
\]
Let $\{p_1, \dots, p_m\}$ be as before.
Then
\[
  \trtost{\shexneigh{\te}}
=
  \stclosed \
  \left\{
  \begin{array}{l}
    \trtetost{\tte} \;\shexeach \\
    \shexinverse{p_1}\ .[0;0] \shexeach \cdots \shexeach \shexinverse{p_m}\
    .[0;0]
  \end{array}
  \right\}
\]
As before, if $\tte = \varepsilon$, then the term $\trtetost{\tte}$ must be
omitted.

This concludes the demonstration of Claim~\ref{claim:app-shex}.

\subsection{Comparative expressiveness of ShEx and SHACL}

In Section~\ref{app:indistinguishabilitySHACL} we show a property expressible in ShEx but not in SHACL, while in Section~\ref{app:indistinguishabilityShEx} we show a property expressible in SHACL but not in ShEx.

\subsubsection{Indistinguishability by SHACL}
\label{app:indistinguishabilitySHACL}

\todo[inline]{Iovka: Introductory sentence for the section.}

\ognjen{Overall, it appears correct.
Most of the cases by induction are claimed to be straightforward, without proving. Perhaps, it can be simplified since it proves a slightly bigger property here than needed... or we define it as an auxiliary lemma.}
\begin{proposition}\label{prop:shacl-inexp}
  The ShEx schema $\shexschema^{eq}$ from Example~\ref{ex:SheXCounting}
  cannot be expressed in SHACL, i.e.\,there is no  SHACL schema $\SHACLSchema'$ such that $\graph$ is valid w.r.t.\,$\SHACLSchema'$ iff $\graph$ is valid w.r.t.\, $\shexschema^{eq}$, for any graph $\graph$. 
\end{proposition}
\begin{proof}
\ognjen{I would provide more guidance wrt to the proof structure, and the idea in general}
To prove this proposition, we first need some preparations. For a node $c$ and an integer $n>0$, a \emph{$(c,n)$-neighbourhood} is a graph $\graph=\{(c,p_1,v_1),\ldots,(c,p_k,v_k)\}$ such that 
\begin{enumerate}
    \item $c,v_1,\ldots,v_k$ are distinct nodes,
    \item for every property $p$, either $p$ does not occur in $\graph$ or $p$ occurs in at least $n$ triples of $\graph$. 
  \end{enumerate}
We say two $(c,n)$-neighbourhoods $\graph_1,\graph_2$ are \emph{similar}, if exactly the same properties appear in $\graph_1$ and $\graph_2$. Intuitively, if $\graph_1,\graph_2$ are similar, then for all properties $p$  we have that either (1) $p$ does not occur in $\graph_1$ nor in  $\graph_2$, or (2) in both $\graph_1$ and $\graph_2$ the node $c$ has at least $n$ outgoing $p$-edges.

Assume a SHACL schema $\SHACLSchema$. We assume that $\SHACLSchema$ does not use expressions of the form $\leqn{n}{\pathExpr}{\varphi}$. This can be assumed w.l.o.g.\,since $\leqn{n}{\pathExpr}{\varphi}$  can be written as $\neg \geqn{n+1}{\pathExpr}{\varphi}$. Let $k$ be the maximal integer that appears among the numeric restrictions of the form  $\geqn{n}{\pathExpr}{\varphi}$ and $\leqn{n}{\pathExpr}{\varphi}$ in $\SHACLSchema$\todo{$\le n$ not in $\SHACLSchema$?}. Assume we have two 
$(c,k+1)$-neighbourhoods $\graph_1,\graph_2$ such that (a) $\graph_1,\graph_2$ are similar, and (b) 
for all nodes $d$ that appear in $\varphi$,\ognjen{fir any node $d$?}  we have that \ognjen{either?} (i) $d=c$, or (ii)  $d$ does not occur in $\graph_1$ or in  $\graph_2$. Then we have ($\dagger$) $\graph_1$ validates w.r.t.\,$\SHACLSchema$ iff $\graph_2$ validates w.r.t.\,$\SHACLSchema$. 
\ognjen{maybe I missed, validates wrt to a node $c$ or some target} 
\todo[inline]{Iovka: "all nodes $d$ that appear in $\varphi$ ?}

To see the above claim, take the binary relation as follows:
\[R=\{(c,c)\}\cup \{(u,v)\mid \exists p: (c,p,u)\in \graph_1,(c,p,v)\in \graph_2 \}.\] 
\ognjen{Not important, but 
I would use something else than $R$, maybe $\sim$ or $\approx$}
We can show  that the following holds for all shape expressions $\varphi$ and property paths $\pi$ that appear in $\SHACLSchema$:
\begin{enumerate}
    \item $\graph_1,d_1\models \varphi$ iff $\graph_2,d_2\models \varphi$, for all $(d_1,d_2)\in R$, and 
    \item $(d_1,d_1')\in \sem{\pi}^{\graph_1}$ iff $(d_2,d_2')\in \sem{\pi}^{\graph_2}$, for all $(d_1,d_2),(d_1',d_2')\in R$.
\end{enumerate}
\ognjen{This is a more general property than we need for the proof. I am wondering if we focus on the exact graph $\G$ defined below we could simplify a bit}
Note that the claim ($\dagger$) follows from (1) above as a special case: since $(c,c)\in R$, we get that $\graph_1,c\models \varphi$ iff $\graph_2,c\models \varphi$. The claims (1-2) are shown by induction on the structure of $\varphi$ and $\pi$. 

We start with the claim (1). Assume arbitrary $(d_1,d_2),(d_1',d_2')\in R$, and consider the possible cases for $\pi$:
\begin{enumerate}[(i)]
    \item $\pi=p$ for some property $p$. Assume  $(d_1,d_1')\in \sem{p}^{\graph_1}$. Then $d_1=c$ and $(c,p,d_1')\in \graph_1$. Since $(d_1',d_2')\in R$, we have that $(c,p,d_2')\in \graph_2$. By the definition of $R$, $d_2=c$ and thus $(d_2,d_2')\in\sem{p}^{\graph_2}$. The other direction is symmetric.
    \item $\pi=k$ for some key $k$. Then trivially $\sem{k}^{\graph_1}=\sem{k}^{\graph_2}=\emptyset$  by the definition of $(c,n)$-neighborhoods and the claim follows.
   
   \item The remaining cases for $\pi=\pi_1^{-}$, $\pi=\pi_1\cdot \pi_2$, $\pi=\pi_1\cup \pi_2$, and $\pi=\pi_1^{*}$ are shown by a straightforward application of the induction hypothesis and the semantics of the operators four operators. 
\end{enumerate}

We can now proceed to prove claim (2). \ognjen{(1)?} Assume arbitrary $(d_1,d_2)\in R$. We only show that  $\graph_1,d_1\models \varphi$ implies $\graph_2,d_2\models \varphi$. The other direction is symmetric. The proof is by structural induction on $\varphi$. Assume $\graph_1,d_1\models \varphi$, and consider the possible cases for $\varphi$:
\begin{enumerate}[(a)]
  \item $\varphi = \geqn{n}{\pathExpr}{\varphi_1}$. Assume  $\graph_1,d_1\models \varphi$. Take the set $F=\{b\mid (d_1,b) \in \iexpr{\pathExpr}{\graph} \land\graph_1,b\models\varphi_1 \}$. There can be two cases: $F=\{c\}$ and $F\neq \{c\}$.

  \medskip
  Suppose $F=\{c\}$. Thus $n=1$ and $\graph_1,c\models\varphi_1$. Since $(c,c)\in R$ and by the induction hypothesis, we get $\graph_2,c\models\varphi_1 $. Moreover, given $(d_1,d_2)\in R$, from  $(d_1,c)\in \sem{\pi}^{\graph_1}$ we infer $(d_2,c)\in \sem{\pi}^{\graph_2}$. Thus we get   
   $\graph_2,d_2\models \geqn{n}{\pathExpr}{\varphi_1}$.

\medskip
  Suppose $F\neq \{c\}$. Since $|F|>0$, there is some $e\in F$ and a unique property $p$ such that $(d_1,p,e)\in \graph_1$. Since $\graph_2$ is a $(c,k+1)$-neigborhood similar to $\graph_1$, we have that $\graph_2$ has $k+1$ distinct edges $(c,p,e_1),\ldots,(c,p,e_{k+1})$ with $(e,e_1),\ldots,(e,e_{k+1})\in R$. Note that $n<k+1$. Since $\graph_1,e\models\varphi_1$, using the induction hypothesis we get that $\graph_2,e_j\models\varphi_1$ for all $1 \leq j \leq k+1$. Moreover, from $(d_1,e)\in \sem{\pi}^{\graph_1}$ we get that  
$(d_2,e_j)\in \sem{\pi}^{\graph_2}$ for all  $1 \leq j \leq k+1$. Thus we get   
   $\graph_2,d_2\models \geqn{n}{\pathExpr}{\varphi_1}$.  
 \item The remaining cases are straightforward.  
\end{enumerate}

%The only interesting aspect is that for all $(d_1,d_2)\in R$ and all property paths $\pi$ in $\SHACLSchema$ we have $\{(d_1',d_2')\mid (d_1,d_1')\in \sem{\pi}^{\graph_1},(d_2,d_2')\in \sem{\pi}^{\graph_2}\} \subseteq R$.

%$\{u\mid (d_1,u)\in \sem{\pi}^{\graph_1}\}$

We can now come back to the main claim of the proposition. %  Take the ShEx schema  $\shexschema$ consisting of the following statement: 
% \[ \shexhasvalue(c) \Rightarrow \shexneigh{(p\shexeach q)^{+}}\]
Towards a contradiction, suppose that there exists a SHACL schema $\SHACLSchema'$ such that $\graph$ is valid w.r.t.\,$\SHACLSchema'$ iff $\graph$ is valid w.r.t.\,$\shexschema^{eq}$, for any graph $\graph$. 

Let $k$ be the maximal integer that appears among the numeric restrictions of the form  $\geqn{n}{\pathExpr}{\varphi}$ and $\leqn{n}{\pathExpr}{\varphi}$ in $\SHACLSchema'$.

Take the graph
\[\graph=\{(c,p,v_j),(c,q,w_j)\mid 1 \leq j \leq k+1\},\]
where none of $v_j$ and $w_j$ appear in $\SHACLSchema'$.
Note that here $\graph$ is such that  $c$ has exactly the same number (i.e.,\,$k+1$) $p$-edges and $q$-edges. Clearly, $\graph$ validates w.r.t.\,$\shexschema$, and hence also $\graph$ validates w.r.t.\,$\SHACLSchema'$

Consider a new graph $\graph'=\graph\cup \{(c,p,u) \}$, where $u$ does not appear in $\graph $. We have that the node $c$ in  $\graph'$ has more outgoing $p$-edges that the number of outgoing $q$-edges, and thus $\graph'$ does not validate w.r.t.\,$\shexschema$. Observe that $\graph$ and $\graph'$ are $(c,k+1)$-neighbourhoods that are similar in the sense defined above, and thus due to ($\dagger$), we have that $\graph'$ does validate w.r.t.\,$\SHACLSchema'$. Contradiction. 
   \end{proof}

\subsubsection{Indistinguishably by ShEx}
\label{app:indistinguishabilityShEx}

%For a graph $\graph$ and edge $e$ in $\graph$, let $\copyswap(\graph, e)$ be the graph obtained by copying $\graph$, then swapping the targets of edge $e$ and its copy $e'$. For instance, in Fig.~\ref{fig:example-shex-counts-edges}, the graph on the right is the copy-swap of the graph on the left and its $\exaccess$ edge (where each node and its copy are horizontally aligned). We can then show that $\graph$ and $\copyswap(\graph, e)$ satisfy the same ShEx schemas (see Appendix~\ref{app:indistinguishabilityShEx}). 

The two graphs in Figure~\ref{fig:example-shex-counts-edges} cannot be distinguished by a ShEx schema.
\todo[inline]{Iovka: adapt to how this is treated in the paper.}
In fact, we show a stronger property. 
Let $\graph = (E, \rho)$ be a graph and $e = (u, p, v) \in E$.
A \emph{double} of $\graph$ is a graph of the form $\graph \cup \graph'$ together with a bijection $d: \nodes(\graph) \to \nodes(\graph')$, where $\graph' = (E', \rho')$ is a disjoint copy of $\graph$.
Now, let $\graph \cup \graph'$ as above be a double of $\graph$, with bijection $d$.
Then $\copyswap(\graph, e)$ is the graph $(E'', \rho \cup \rho')$ such that 
$$E'' = E \cup E' \setminus \{e, d(e)\} \cup \left\{(u, p, d(v)), (d(u), p, v)\right\}.$$

Back to the graphs in Fig.~\ref{fig:example-shex-counts-edges}, we have $\graph_{\text{right}} = \copyswap(\graph_{\text{left}}, e)$, where $e$ is the unique edge in $\graph_{\text{left}}$ labelled $\exaccess$, and $\graph_{\text{left}}$, resp. $\graph_{\text{right}}$, are the graphs on the left, resp. on the right, in Fig.~\ref{fig:example-shex-counts-edges}.

\begin{lemma}
  For every ShEx schema $\shexschema$, every graph $\graph$ and every edge $e$ in $\graph$, if $\graph \models \shexschema$, then
  $\copyswap(\graph, e) \models \shexschema$.
\end{lemma}
\begin{proof}
Let $e = (u, p, v)$ and $\graph' = \copyswap(\graph, e)$.
The proof easily follows by structural induction on shape expression
$\varphi$, where the induction base $\graph, u \models \se$ \iff
$\graph', u \models \se$ \iff $\graph', d(u) \models \se$  is
immediate due to the definition of the $\copyswap$ function.
% \begin{itemize}[\textbullet]
% \item
%   consider a derivations for an atomic statement $\graph, v \models \se$;
% \item
%   show that for every node $u$ in $\graph$ and every shape expression $\se$, if
%   $\graph, u \models \se$ then there exist derivations for $\graph',
%   d_1(u) \models \se$ and for $\graph', d_2(u) \models \se$. Such derivations
%   are basically copies of the derivation for $\graph, u \models \se$ in which,
%   if $e$ is part of the derivation, we need to switch from $d_1(x)$ to $d_2(x)$
%   and vice versa.
% \end{itemize}
\end{proof}
%\input{Appendix-Shex-garbage-collector}

%%% Local Variables:
%%% mode: latex
%%% TeX-master: "Article"  
%%% End:                                                                                                                         

%!TEX root =  Article.tex

\section{Standard PG-Schema}
\label{sec:standard-pg-schema}

The version of PG-Schema presented in the body of the paper is a variant of PG-Schema that is constructed to preserve the essence of the original PG-Schema as presented in \cite{ABDF23} but also to fit a paradigm of a shape-based schema language like SHCAL and ShEx, and in particular to follow the paradigm where a schema consists of a set of selector-shape pairs. 
However, the original PG-Schema follows a different paradigm, namely one where a schema consists of a set of node types, a set of edge types, and a set of constraints.
To show that nevertheless the version of PG-Schema in the body of the paper preserves its core functionality, we will present here an intermediate version that we will refer to as \emph{PG-Schema on Common Graphs} while we refer to the version of PG-Schema in the body of the paper as \emph{shape-based PG-Schema}, and to the PG-Schema defined in \cite{ABDF23} as \emph{original PG-Schema}.

\subsection{PG-Schema on Common Graphs}

The central idea of the original PG-Schema in \cite{ABDF23} 
is that a schema (called \emph{graph type} 
in this context) consists of three parts: (1) a set of node types, 
(2) a set of edge types, and (3) a set of graph constraints
that represents logical statements about the property graph that must hold for it to be valid.  
A particular property graph is then said to be valid wrt.\ such a graph type if (1) every node 
in the property graph is in the semantics of at least one node type, (2) every edge in the property graph 
is in the semantics of at least one edge type, and (3) the property graph satisfies all specified graph constraints.

The organisation of this section is as follows. We first discuss the notions of \emph{node types} and \emph{edge types}.
After that we discuss how path expressions are defined, after which we discuss what graph constraints look like in this setting.
In the final two subsections we discuss how this version of PG-Schema relates the original defined in
\cite{ABDF23}, and how it relates to the one define in this paper.

\subsubsection{Node types}

The purpose of node types in the original PG-Schema is to describe nodes, their properties and their labels. 
Since in the common graph model nodes there are no labels, node types become simply record types
where the record fields describe the allowed keys.
Therefore node types are here defined to be the same as the content types defined in Definition~\ref{def:contentType}.
In the original PG-Schema it was possible to indicate if these record types are \emph{closed} or \emph{open},
where the former indicates that only the indicated keys are allowed, and the latter that additional keys
are allowed. This is easily expressed with such node types, and for example a node type that requires the presence of a key with name $\excard$ and a value of type $\mathbbm{int}$, and allows in addition other keys, is represented by $\closedRT{\excard : \mathbbm{int}} \tAnd \top$.

\subsubsection{Edge types}

In the original PG-Schema there is a notion of edge type, which consists of three parts: 
(1) a type describing the source node, (2) a type describing describing the contents of the edge itself, and (3) a type describing the the target node.
Since in common graphs the content of an edge is just a label, a type describing this content can be simply an expression of the form $\pwc$ (a wild-card indicating that any label is possible) or a finite set $P$ of labels (indicating that only these labels are allowed). 
So we get the following definition for edge types.

\begin{definition}[Edge type]
An \emph{edge type} is an expression $\etype$ of the form defined by the grammar
$$\etype \gDef \et{\ntype}{\pwc}{\ntype} \gMid \et{\ntype}{P}{\ntype} \gMid  \etype \tAnd \etype  \gMid  \etype \tOr \etype  \gEnd $$
where $P$ is a finite subset of $\Predicates$.
\end{definition}

As for node types, we define for edge types a value semantics, which in this case defines which combinations of (1) source node content, (2) edge content, and (3) target node content are allowed. 

\begin{definition}[Value semantics of edge types]
 With an \emph{edge type} $\etype$ we associate a \emph{value semantics} $\sem{\etype} \subseteq \Records \times \Predicates \times \Records$ which is defined with induction on the structure of $\etype$ as follows:
 \begin{enumerate}
   \item $\sem{\et{\ntype_1}{\pwc}{\ntype_2}} = \sem{\ntype_1} \times \Predicates \times \sem{\ntype_2}$
   \item $\sem{\et{\ntype_1}{P}{\ntype_2}} = \sem{\ntype_1} \times P \times \sem{\ntype_2}$   
   \item \begin{tabbing}
         $\sem{ \etype_1 \tAnd \etype_2 } = \{ $ \= $ ( (r_1 \cup s_1), p, (r_2 \cup s_2) ) \in \Records \times \Predicates \times \Records \mid $ \\
         \> $( r_1, p, r_2 ) \in \sem{\etype_1} \land ( s_1, p, s_2 ) \in \sem{\etype_2} \}$
      \end{tabbing}
   	\item $\sem{ \etype_1 \tOr \etype_2 } = \sem{\etype_1} \cup \sem{\etype_2}$
 \end{enumerate}
\end{definition}

\subsubsection{Path expressions}

We define here a notion of path expression that we call \emph{extended PG-path expression} and that is similar to the notion of PG-path expression of Definition~\ref{def:pgpaths-syntax}, except that in the positions where a content type $\ntype$ is allowed, 
we also allow an edge type $\etype$.

\begin{definition}[Extended PG-path expressions] 
\label{def:ext-pgpaths-syntax}
An extended  PG-path expression is an expression $\pexpr$ of the form defined by the  grammar
\begin{align*} 
& \pexpr \gDef  \ppexpr \gMid \ppexpr \cdot k \gMid k^{-} \cdot \ppexpr \gMid k^{-} \cdot \ppexpr \cdot  k' \gEnd \\
& \ppexpr \gDef \keyIsVal{k}{c} \gMid \neg \keyIsVal{k}{c} \gMid \ntype \gMid \lnot \ntype \gMid \etype \gMid \lnot \etype \gMid \\
& \quad \quad p \gMid \lnot P \gMid  
{\ppexpr}^{-} \gMid \ppexpr \cdot \ppexpr \gMid \ppexpr \cup \ppexpr \gMid {\ppexpr}^{*} \gEnd
\end{align*}
where $k,k' \in \Keys$, $c \in \Values$, $\ntype$ is a content type, $p \in \Predicates$, and $P \subseteq_{\mathit{fin}} \Predicates$. 
% We use $k$, $k^{-}$, and $k^{-}\cdot k'$ as short-hands for PG-path expressions $\top\cdot k$, $k^{-}\cdot \top$, and $k^{-}\cdot \top\cdot  k'$, respectively. 
\end{definition}

The semantics of extended PG-path expressions is identical to that of PG-path expressions for the expressions they have in common, and for the additional parts, the edge types $\etype$ and $\lnot \etype$, the semantics is given in Table~\ref{tab:semSPGtypes}.

\begin{table}[tb]
  \caption{Semantics extended PG-path expressions.}
  \label{tab:semSPGtypes}  
  \centering
  \begin{tabular}{cl}
    \toprule
    $\pexpr$ & $\gsem{\pexpr}\subseteq (\Nodes\cup\Values)\times (\Nodes\cup\Values)\ $  for  $\graph = (E, \rho)$ \\[2pt]
    \midrule    
    $\etype$ & $\left\{ (u, v) \mid \exists p : (u, p, v) \in E \land (\rho(u), p, \rho(v)) \in \sem{\etype} \right\}$ \\
    $\lnot \etype$ & $\left\{ (u, v) \mid \exists p : (u, p, v) \in E \land (\rho(u), p, \rho(v)) \notin \sem{\etype} \right\}$ \\
    \bottomrule
  \end{tabular}
\end{table}

\subsubsection{Graph constraints}

The graph constraints in the original PG-Schema are based on the constraints discussed in PG-Keys \cite{ABDF21}. 
Although the latter paper focuses on key constraints, it also discusses other closely related cardinality constraints. 
We capture these constraints here in the context of the common graph data model with the following formal definition.

\begin{definition}[PG-constraint]
 A \emph{PG-constraint} is a formula of one of the following three forms:
\begin{description}
    \item[Key:] $\forall x : \varphi(x) \Lleftarrow \Key \, \bar{y} : \psi(x, \bar{y})$
    \item[Upb:] $\forall x : \varphi(x) \to \exists^{\leq n} \, \bar{y} : \psi(x, \bar{y})$
    \item[Lwb:] $\forall x : \varphi(x) \to \exists^{\geq n} \, \bar{y} : \psi(x, \bar{y})$
\end{description}
where $x$ is a variable that ranges over nodes and values, $\varphi(x)$ and $\psi(x, \bar{y})$ are formulas of the form $\exists \bar{z} : \xi$ with $\bar{z}$ a vector of node and value variables and $\xi$ a conjunction of atoms of the form $\pathExpr(z_i, z_j)$ with $z_i$ and $z_j$ either equal to $x$, or in $\bar{y},$ or in $\bar{z}$, and $\pathExpr$ an extended PG-path expression, such that the free variables in $\varphi(x)$ are just $x$ and those in $\psi(x, \bar{y})$ are $x$ and the variables in $\bar{y}$. 
 \end{definition}

The semantics of the constraints of the form \textbf{Key} is the logical formula $\forall \bar{y} : \exists^{\leq 1} x : \varphi(x) \land \psi(x, \bar{y})$. This corresponds to the intuition that it defines a key constraint for all values or nodes $x$ that are selected by $\varphi(x)$ and for those it specifies that that the vector $\bar{y}$ for which $\psi(x, \bar{y})$ holds identifies at most one such $x$.
So the symbol $\Lleftarrow$ should be read here as stating that the left-hand side is functionally determined by the right-hand side.

For the constraints of the forms \textbf{Upb} and \textbf{Lwb} the interpretation is simply the usual one in first-order logic.

\subsection{Comparison with the original PG-Schema }

% \todo[inline]{Let's see how much of it we really want to keep here. Goal: keep it roughly analogous to the discussion for RDF-based formalisms. @Jan @Filip}
% \todo[inline]{JH: This has not been done yet.}

The PG-Schema on Common Graphs defined here introduces two important simplification w.r.t.\ original PG-Schema: (1) It is defined over common graphs which simplifies the property graph data model in several ways and (2) it assumes what is called the STRICT semantics of a graph type in  \cite{ABDF23} and ignores the LOOSE semantics. We briefly discuss here why these simplification preserve the essential characteristics of the original schema language.

\subsubsection{Concerning the simplification of the data model}

As discussed in Section~\ref{sect:PGCGComparison} common graphs simplify property graphs in three ways: (1) nodes only have properties and no labels, (2) edges only have one label and no properties, and (3) edges have no independent identity. 
However, these features can be readily simulated in the common graph data model.
For example, edges with identity can be simulated by nodes that have an outgoing edge with label $\Exkey{source}$ and an outgoing edge with label $\Exkey{target}$ to respectively the source node and the target node of the simulated edge.
Moreover, labels can be simulated by introducing a special dummy value $\Lambda$ that is used for keys that represent labels. 
For example, a node $n$ where $\rho(n)$ contains the pairs $(\Exkey{Person}, \Lambda)$, $(\Exkey{Employee}, \Lambda)$, $(\Exkey{hiringDate}, \textit{12-Dec-2023})$, and $(\Exkey{fulltime}, \Exkey{true})$, simulates a node with labels $\Exkey{Person}$ and $\Exkey{Employee}$, and properties $\Exkey{hiringDate}$ and $\Exkey{fulltime}$. 

It is not hard to see how under such a simulation PG-Schema on Common Graph could simulate a more powerful schema language where we could use tests in path expressions for the presence (or absence) of (combinations of) labels in path expressions and
tests for presence (or absence) of (combinations of) properties of edges.
Moreover, we could navigate via simulated edges and test for certain properties with a path expression of the form $\Exkey{source}^{-} \cdot \pathExpr \cdot \Exkey{target}$ where $\pathExpr$ simulates any test over the content of the edge.
Finally, we could straightforwardly simulate key and cardinality constraints for edges.

\subsubsection{Concerning the STRICT and LOOSE semantics} 

In the original PG-Schema there is a separate LOOSE semantics defined for graph types. 
In that case the set of node types and the set of edge types in the graph type are ignored and 
a property graph is said to be already valid wrt.\ a graph type if it satisfies all graph constraints in the graph type.
The LOOSE interpretation can be easily simulated in PG-Schema on Common Graphs
by letting the set of node types contain only $\top$, the trivial node type, 
and the set of edge types contain only $\et{\top}{\pwc}{\top}$, the trivial edge type.

\subsection{Comparison with Shape-based PG-Schema}

The constraints of the forms \textbf{Upb} and \textbf{Lwb} are very similar to the selector-shape pairs presented for 
PG-Schema in Section~\ref{sec:pgschema-simplified}. 
Indeed, the selector is in this case the formula $\varphi(x)$ and the shape is 
the formulas of the forms $\exists^{\leq n} \bar{y} : \psi(x, \bar{y})$ and $\exists^{\geq n} \bar{y} : \psi(x, \bar{y})$.
However, there are also several notable differences:
(1) The schema in Shape-based PG-Schema only consists of constraints and does not separately define sets of allowed node and edge types.
(2) There are no edge types in path expressions. 
(3) All constraints are restricted so that $\bar{y}$ is just a single variable.
(4) There are no constraints of the form \textbf{Key}.
(5) The constraints are syntactically restricted such that (a)
$\varphi(x)$ is restricted to just one atom, so the form $\exists z : \pathExpr(x, z)$, and (b)
$\psi(x, \bar{y})$ is restricted to just one atom, so the form $\pathExpr(x, y)$. It is this last restriction that allows a notation in description-logics style without variables.
Apart from these restrictions, there is also a generalisation, namely in Section~\ref{sec:pgschema-simplified} the shapes are closed under intersection.
That this does not change the expressive power is easy to see, since a selector-shape pair of the form $(\sel, (\varphi_1 \land \varphi_2))$ 
can always be replaced with the combination of the pairs  $(\sel, \varphi_1)$ and $(\sel, \varphi_2)$ without changing the semantics of the schema.

In the following subsections we discuss the previously mentioned restrictions.

\subsubsection{No separate sets of allowed node types and edge types}

It is not hard to show that this can be simulated. 
Assume for example we have a graph type with a set of node types $\{ \ntype_1, \ntype_2, \ntype_3 \}$.
The check that each node must be in the semantics of at least one of these node types can be simulated in PG-Schema on Common Graphs by 
the \textbf{Lwb} constraint $$\forall x : \top(x, x) \to \exists y : (\ntype_1 \tOr \ntype_2 \tOr \ntype_3)(x, y)$$ 
Note that node types are closed under the $\tOr$ operator, and so $(\ntype_1 \tOr \ntype_2 \tOr \ntype_3)$ is indeed a node type, and therefore an extended PG-Path expression in PG-Schema on Common Graphs.
Recall that a node type acts in a path expression as the identity relation restricted to nodes that are in the semantics of that type.

Similarly, if the set of edge types of a graph type is $\{ \etype_1, \etype_2, \etype_3 \}$, we can ensure that each edge is in the semantics of at least one of these edge types using the following \textbf{Upb} constraint in PG-Schema on Common Graphs: 
$$\forall x : \top(x, x) \to \exists^{\leq 0} y : \lnot(\etype_1 \tOr \etype_2 \tOr \etype_3)(x, y)\,.$$
Note that edge types are closed under union, and so $\etype_1 \tOr \etype_2 \tOr \etype_3$ is also an edge type, and in addition edge type can appear negated and extended PG-Path expressions, and so $\lnot(\etype_1 \tOr \etype_2 \tOr \etype_3)$ is indeed a valid path expression in PG-Schema on Common Graphs.

\subsubsection{No edge types in path expressions}

It is not hard to show that path expressions that contain tests involving edge types 
can be rewritten to equivalent path expressions that do not use edge types.

We first consider the non-negated edge types in path expressions. We start with the observation that we can normalise edge types to a union of edge types that do not contain the $\tOr$ operator.
This is based on the following equivalences for path semantics that allow us to push down the $\tOr$ operator:
\begin{itemize}
    \item $\et{(\ntype_1 \tOr \ntype_2)}{\alpha}{\ntype_3} \equiv ( \et{\ntype_1}{\alpha}{\ntype_3} \tOr \et{\ntype_2}{\alpha}{\ntype_3} )$
    \item $\et{\ntype_1}{\alpha}{(\ntype_2 \tOr \ntype_3)} \equiv ( \et{\ntype_1}{\alpha}{\ntype_2} \tOr \et{\ntype_1}{\alpha}{\ntype_3} )$
\end{itemize}
In a next normalisation step we can remove bottom-up the $\tAnd$ operator using the following rules, where we use the symbol $\etype_{\emptyset}$ to denote the empty edge type:
\begin{itemize}
    \item $(\et{\ntype_1}{\alpha}{\ntype_2}) \tAnd (\et{\ntype_3}{\beta}{\ntype_4}) \equiv \et{(\ntype_1 \tAnd \ntype_3)}{\alpha \sqcap \beta}{(\ntype_3 \tAnd \ntype_4)}$
\end{itemize}
where $\sqcap$ is defined such that (1) $\pwc \sqcap P = P \sqcap \pwc = P$ for $P \subseteq \Predicates$, and (2) $P \sqcap Q = P \cap Q$ for $P, Q \subseteq \Predicates$.

As a final normalisation step we get rid of edge types $\et{\ntype_1}{P}{\ntype_2}$ where $P$ contains two or more predicates, by applying the rule:
\begin{itemize}
    \item $\et{\ntype_1}{\{ p_1, \ldots, p_k \}}{\ntype_2} \equiv (\et{\ntype_1}{\{ p_1 \}}{\ntype_2} \tOr \ldots \tOr \et{\ntype_1}{\{ p_k \}}{\ntype_2})$
\end{itemize}

After these normalisation steps we will have rewritten the edge type to the form $(\etype_1 \tOr \ldots \tOr \etype_k)$ with each $\etype_i$ a \emph{primitive edge type} in the sense that it cannot be normalised further and therefore one of the following forms:
(1) $\et{\ntype_1}{\pwc}{\ntype_2}$, 
(2) $\et{\ntype_1}{\{ p \}}{\ntype_2}$, and 
(3) $\et{\ntype_1}{\emptyset}{\ntype_2}$.
We can express such an edge type $(\etype_1 \tOr \ldots \tOr \etype_k)$ as a path expression $(\pathExpr_1 \cup \ldots \cup \pathExpr_k)$, where each $\pathExpr_i$ is constructed as follows:
\begin{itemize}
    \item $\et{\ntype_1}{\pwc}{\ntype_2} \equiv \ntype_1 \cdot \neg \emptyset \cdot \ntype_2$
    \item $\et{\ntype_1}{\{ p \}}{\ntype_2} \equiv \ntype_1 \cdot p \cdot \ntype_2$
    \item $\et{\ntype_1}{\emptyset}{\ntype_2} \equiv \neg \top$
\end{itemize}
Recall that $\neg\top$ is the negation of the trivial node type and so in a path expression represents the empty binary relation.

We now turn our attention to negated edge types. The part under the negation can be normalised as before,
and so we end up with an edge type of the form $\neg (\etype_1 \tOr \ldots \tOr \etype_k)$ with each $\etype_i$ a primitive edge type.
This can be represented as a path expression $(\pathExpr_1 \cup \ldots \cup \pathExpr_{m})$ where each $\pathExpr_j$ is a path expression the expresses a particular reason that an edge might not conform to any of the types in $\etype_1, \ldots, \etype_k$.
To illustrate this consider as an example the following negated edge type:
\[\lnot \big( \et{\ntype_1}{\{ p \}}{\ntype_2} \tOr \et{\ntype'_1}{\{ p' \}}{\ntype'_2} \big)\]
This can be simulated in a path expression by replacing it with the following path expression:
\[\begin{array}{c@{\hspace{6pt}}c@{\hspace{6pt}}c@{\hspace{6pt}}c@{\hspace{6pt}}c@{\hspace{6pt}}c@{\hspace{6pt}}c}
& \lnot\ntype_1\cdot \lnot\ntype'_1 \cdot \lnot\emptyset & \cup &
 \lnot\ntype_1 \cdot \lnot\{p'\} & \cup &
 \lnot\ntype_1 \cdot \lnot\emptyset \cdot \lnot\ntype'_2 & \cup\\
\cup\;  & \lnot\ntype'_1 \cdot \lnot\{p\} & \cup &
\lnot\{p,p'\} & \cup &
  \lnot\{p\} \cdot \lnot\ntype'_2 & \cup\\
\cup \; & \lnot\ntype'_1\cdot  \lnot\emptyset \cdot \lnot\ntype_2 &\cup&
\lnot\{p'\} \cdot \lnot\ntype_2  &\cup &\;
\lnot\emptyset \cdot \lnot\ntype_2 \cdot \lnot\ntype'_2\;
\end{array}\]
Note that this indeed enumerates all the ways that an edge could not be in the semantics of $\big( \et{\ntype_1}{\{ p \}}{\ntype_2} \tOr \et{\ntype'_1}{\{ p' \}}{\ntype'_2} \big)$. Basically we pick for each of the primitive edge types whether the edge is not in the semantics because of (1) the source node, (2) the label, or (3) the target node.

\subsubsection{Only single variable counting}

The restriction to allow only one variable in $\bar{y}$ is introduced because in SHACL and ShEx all the counting is also restricted to single values and nodes, rather than tuples of values and nodes. 
Although this is often useful in real-world data modelling, e.g., to represent composite keys, this restriction is introduced to make PG-Schema more comparable to SHACL and ShEx.

\subsubsection{No \textbf{Key} constraints}
If \textbf{Key} constraint are restricted to single-variable counting, \textbf{Key} constraints are of the form \[\forall x : \varphi(x) \Lleftarrow \Key \, y : \psi(x, y)\,.\] 
Recall that its semantics is defined by the formula $\forall y : \exists^{\leq 1} x : \varphi(x) \land \psi(x, y)$. If $y$ matches nodes (which can be detected based on path expressions used in the atoms involving $y$), we can equivalently express this constraint in PG-Schema for Common Graphs as
\[
\forall y : \top(y, y) \to \exists^{\leq 1} x: \varphi(x)\land\psi(x,y)
\]

If $y$ matches values, then it is used in the first position of an atom whose path expressions begins from $k^{-}$ or in the second position of an atom whose path expression ends with $k$. In either case, we can equivalently express this constraint in PG-Schema for Common Graphs as
\[
\forall y : (k^{-}\!\cdot k) (y', y) \to \exists^{\leq 1} x: \varphi(x)\land\psi(x,y)
\]

\subsubsection{Only one atom in formulas}

This restriction of the query language underlying PG-Schema for Common Graphs limits the expressive power of PG-Schema, but similar restrictions are present in SHACL and ShEx. Some additional expressive power could be gained by allowing tree-shaped conjunctions of atoms with at most 2 free  variables, but this would further complicate the formal development.

\section{More on the core}
\label{sec:appendix-core}

In this section we prove Proposition~\ref{prop:core}. Recall that common shapes are defined by the grammar 
\begin{align*}
\varphi  \gDef  & 
 \exists\,\pexpr
\gMid \exists^{\leq n} \, \pexpr_1
\gMid \exists^{\geq n} \, \pexpr_1 \gMid 
\exists\, \ntype \land \not\exists\, \lnot P
\gMid \varphi \land \varphi \gEnd\\
\ntype \gDef &\closedRT{}\  \gMid\  \closedRT{k : \vtype} \ \gMid\   \ntype \tAnd \ntype  \ \gMid \  \ntype \tOr \ntype  \gEnd\\
\pexpr_0 \gDef & 
\keyIsVal{k}{c} \gMid 
\lnot \keyIsVal{k}{c} \gMid 
\ntype\tAnd\top\gMid \lnot (\ntype\tAnd\top) \gMid \pexpr_0 \cdot \pexpr_0 \gEnd\\
\pexpr_1 \gDef &  \pexpr_0  \cdot p \cdot
\pexpr_0 
\gMid  \pexpr_0  \cdot p^{-} \cdot
\pexpr_0  \gMid \pexpr_0\cdot k \gMid k^{-}\cdot \pexpr_0 \gEnd\\
\ppexpr \gDef & \pexpr_0  \gMid p 
\gMid {\ppexpr}^{-} \gMid \ppexpr \cdot \ppexpr 
\gMid \ppexpr \cup \ppexpr \gEnd \\
\pexpr \gDef & \ppexpr \gMid \ppexpr \cdot k \gMid k^{-}\cdot \ppexpr \gMid k^{-}\cdot \ppexpr\cdot k' \gEnd
\end{align*}
where $n \in \mathbb{N}$, $P\subseteq_{\mathit{fin}} \Predicates$, $k,k'\in\Keys$, $c\in\Values$, and $p\in\Predicates$. We will refer to  PG-path expressions defined by the nonterminal $\pexpr_0$ in the grammar as  \emph{filters}.

%First, note that  $\exists^{\geq 0} \pexpr_0\cdot k$ is trivially satisfied in each node and  $\exists^{\geq n} \pexpr_0\cdot k$ is equivalent to $\exists\, \pexpr_0\cdot k$ for all $n>0$. Hence, we can assume that core shapes do not contain subexpressions of the form $\exists^{\geq n} \pexpr_0\cdot k$.

% Second, we can assume without loss of generality that in each core shape of the form $\exists\, \pexpr$, the PG-path expression $\ppexpr$ underlying $\pexpr$ is a union of concatenations of filters and atomic path expressions of the form $p$ or $p^{-}$.

%Third, we can assume without loss of generality that each closed content type has the form $\ntype_1 \tOr \ntype_2 \tOr \dots \tOr \ntype_n$, where each $\ntype_i$ has the form  $\closedRT$ or $\closedRT{k_1:\vtype_1}\tAnd\closedRT{k_2:\vtype_2}\tAnd \dots \tAnd \closedRT{k_m :\vtype_m}$.

%From now on we only consider core schemas normalized as described above. 

The following two subsections describe the translations of common schemas to SHACL and ShEx. The translations are very similar but we include them both for the convenience of the reader.

\subsection{Translation to SHACL}

\todo[inline]{Cem: Explain here that open type is just a short-hand for $\ntype \&  \top$. It could also be mentioned directly in the Core section, that the core language only allows for open types (or arbitrary content types with a ``guard'' to ensure closedness). }

\begin{lemma}
\label{lem:contents-shacl}
For each open content type $\ntype$ there is a SHACL shape $\varphi_{\ntype}$ such that  $\graph, v\models \varphi_{\ntype}$ iff $\rho(v) \in \sem{\ntype}$ for all  $\graph=(E,\rho)$ and $v\in\nodes(\graph)$.
\end{lemma}

\begin{proof}
For the content type $\top$ the corresponding SHACL shape is $\top$.
For a content type of the form \[\closedRT{k_1:\vtype_1}\tAnd\closedRT{k_2:\vtype_2}\tAnd \dots \tAnd \closedRT{k_m :\vtype_m}\tAnd \top\,,\] the corresponding SHACL shape is \[\exists\, k_1. \test(\vtype_1) \land \exists\, k_2. \test(\vtype_2) \land \dots \land \exists\, k_m. \test(\vtype_m)\,.\]

\todo[inline]{Cem: maybe a technical lemma that shows why/how every content type can be represented in the needed \emph{normal form} would make the proofs of this section easier to follow. Feels strange to just leave a (mostly, but not entirely trivial) claim open in a proof.}

Finally, every open content type different from $\top$ can be expressed as 
\[(\ntype_1 \tOr \dots \tOr \ntype_\ell)\tAnd \top\,,\]
where each $\ntype_i$ is a  content type of the form $\closedRT{k_1:\vtype_1}\tAnd\closedRT{k_2:\vtype_2}\tAnd \dots \tAnd \closedRT{k_m :\vtype_m}$ for some $m$.
The corresponding SHACL shape is 
\[\varphi_1 \lor \dots \lor \varphi_\ell\,,\]
where $\varphi_i$ is the SHACL shape corresponding to the content type $\ntype_i \tAnd \top$.
\end{proof}

\begin{lemma}
\label{lem:filter-shacl} For each filter $\pexpr_0$ there is a SHACL shape $\varphi_{\pexpr_0}$ such that  $\graph, v\models \varphi_{\pexpr_0}$ iff $(v,v) \in \sem{\pexpr_0}^\graph$ for all  $\graph$ and $v\in\nodes(\graph)$. 
\end{lemma}

\begin{proof} By Lemma~\ref{lem:contents-shacl}, the claim holds for $\pexpr_0=\ntype\tAnd \top$. For $\{k:c\}$ the corresponding SHACL shape is $\exists k.\hasvalue(c)$. 
As SHACL shapes are closed under negation, the claim holds for $\lnot \{k:c\}$ and $\lnot(\ntype\tAnd \top)$.
Finally, concatenations of filters correspond to conjunctions of shapes, so the claim follows because SHACL shapes are closed under conjunction. 
\end{proof}

\begin{lemma}
\label{lem:paths-shacl}
For each common shape of the form $\exists\,\pexpr$ there is a SHACL shape $\varphi_{\exists\pexpr}$ such that  $\graph, v\models \varphi_{\exists\pexpr}$ iff $\graph,v\models \exists\, \pexpr$ for all $\graph$ and $v\in\nodes(\graph) \cup \values(\graph)$. 
\end{lemma}

\begin{proof}
Let us first look at common shapes of the form $\exists\,\pexpr$ where $\pexpr$ is a concatenation of filters and atomic path expressions of the form  $p$, $p^{-}$, $k$, or $k^{-}$. Without loss of generality we can assume that the concatenation ends with a filter or with $k$.
We proceed by induction on the length of the concatenation. The base cases are $\exists \pi_0$ and $\exists k$, which correspond to $\varphi_{\pi_0}$ (Lemma~\ref{lem:filter-shacl}) and $\exists k.\top$. 
For $\exists\, \pexpr_0\cdot\pexpr$ we can take $\varphi_{\pexpr_0} \land \varphi_{\exists \pi}$. For $\exists\, p\cdot\pexpr$ we can take $\exists p.\varphi_{\exists \pi}$, and similarly for $\exists\, p^{-}\cdot\pexpr$ and $\exists\, k^{-}\cdot\pexpr$. 

The general case follows because SHACL shapes are closed under union. Indeed, because our PG-path expressions are star-free, we can assume without loss of generality that in each common shape of the form $\exists\, \pexpr$, the PG-path expression $\ppexpr$ underlying $\pexpr$ is a union of concatenations of filters and atomic path expressions of the form $p$ or $p^{-}$. Then, for \[\exists\, k^{-}\cdot(\pexpr^1 \cup \dots \cup\pexpr^m)\cdot k'\] we can take \[\varphi_{\exists k^{-}\cdot\pexpr^1\cdot k'} \lor \dots \lor \varphi_{\exists k^{-}\cdot\pexpr^m\cdot k}\,.\] Simiarly for $\exists\, k^{-}\cdot(\pexpr^1 \cup \dots \cup\pexpr^m)$, $\exists\, (\pexpr^1 \cup \dots \cup\pexpr^m)\cdot k'$, and $\exists\, (\pexpr^1 \cup \dots \cup\pexpr^m)$.
\end{proof}

\todo[inline]{Cem: in \Cref{lem:paths-shacl}, it seems strange to me to use natural induction on the length of paths (where the use of IH in the step case is left very implicit) and not opt for the more obvious choice of a structural induction on paths. Space saving measure? }

\begin{lemma}
\label{lem:shapes-shacl}
For each common shape $\varphi$ there is a SHACL shape $\hat \varphi$ such that  $\graph, v\models \varphi$ iff $\graph,v\models \hat \varphi$ for all $\graph$ and $v\in\nodes(\graph)  \cup \values(\graph)$. 
\end{lemma}

\begin{proof} 
Because SHACL shapes are closed under conjunction, it suffices to prove the claim for the atomic common shapes of the forms $\exists\, \pexpr$, $\exists^{\leq n}\pexpr_1$, $\exists^{\geq n}\pexpr_1$, and $\exists \ntype \land \not\exists\lnot P$. The first case follows from Lemma~\ref{lem:paths-shacl}. 

Let us look at common shapes of the form
$\exists^{\geq n}\pexpr_1$. If $n=0$ we can simply take $\top$. Suppose $n>0$. Then, for \[\exists^{\geq n} \pexpr_0\cdot p\cdot \pexpr'_0\] we can take \[\varphi_{\pexpr_0} \land \exists^{\geq n}p. \varphi_{\pexpr'_0},\] and similarly for $\exists^{\geq n}\pexpr_0\cdot p^{-}\cdot \pexpr'_0$, $\exists^{\geq n}\pexpr_0\cdot k^{-}\cdot \pexpr'_0$, and $\exists^{\geq n}\pexpr_0\cdot k$ (using $\top$ instead of $\varphi_{\pexpr'_0}$). 

Next, we consider common shapes of the form 
$\exists^{\leq n}\pexpr_1$. For \[\exists^{\leq n} \pexpr_0\cdot p\cdot \pexpr'_0\] we can take \[\lnot\varphi_{\pexpr_0} \lor \exists^{\leq n}p. \varphi_{\pexpr'_0},\] and similarly for $\exists^{\leq n}\pexpr_0\cdot p^{-}\cdot \pexpr'_0$, $\exists^{\leq n}\pexpr_0\cdot k^{-}\cdot \pexpr'_0$, and $\exists^{\leq n}\pexpr_0\cdot k$ (again, using $\top$ instead of $\varphi_{\pexpr'_0}$).

Finally, let us consider a common shape of the form $\exists\, \ntype \land \not\exists\,\lnot P$. Suppose first that \[\ntype = \closedRT{}\,.\] Then, the corresponding SHACL shape is simply \[\closed(P)\,.\] Next, suppose that \[\ntype = \closedRT{k_1:\vtype_1}\tAnd \dots \tAnd \closedRT{k_m :\vtype_m}\,.\] 
Then, the corresponding SHACL shape is 
\[
\exists k_1.\test(\vtype_1) \land \dots \land \exists k_m.\test(\vtype_m) \land \closed\big(\{k_1, \dots,k_m\} \cup P\big)\,.\]
In general, as in Lemma~\ref{lem:contents-shacl}, we can assume that 
\[\ntype = \ntype_1 \tOr \dots \tOr \ntype_m\] where each $\ntype_i$ is of one of the two forms considered above. The corresponding SHACL shape is then 
\[\varphi_1 \lor \dots \lor \varphi_m\] where $\varphi_i$ is the SHACL shape corresponding to $\exists \ntype_i \land \not\exists\lnot P$, obtained as described above. 
\end{proof}

\begin{lemma}
\label{lem:schemas-shacl}
For every common schema there is an equivalent SHACL schema.  
\end{lemma}

\begin{proof} 
Let $\schema$ be a common schema. We obtain an equivalent SHACL schema $\schema'$ by translating each  $(\sel, \varphi) \in \schema$ to $(\sel', \varphi')$ such that for all $\graph$ and $v\in \Nodes\cup\Values$, 

\begin{center}
 $\graph, v \models \sel$ implies $\graph, v \models \varphi$ 
 
 iff 
 
 $\graph, v \models \sel'$ implies $\graph, v \models \varphi'$.
 \end{center}
 Recall that $\sel$ is a common shape of one of the following forms:
\[ 
\exists\, k \,,\;
\exists\, p \cdot \pexpr\,, \;
\exists\, p^{-}\!\cdot \pexpr \,,\;
\exists\, \{k:c\}\cdot \pexpr \,, \;
\exists\, \big(\{k:\vtype\}\tAnd\top\big)\cdot \pexpr \,, \;
\exists\, k^{-}\!\cdot \pexpr\,.
\] 
For $\sel'$ we take, respectively, 
\[ 
\exists\, k.\top \,,\quad
\exists\, p.\top\,, \quad 
\exists\, p^{-}.\top \,,\quad 
\exists\, k.\top\,, \quad 
\exists\, k.\top\,, \quad 
\exists\, k^{-}.\top\,.
\]
For $\varphi'$ we take $\lnot \varphi_{\sel} \lor \hat\varphi$ where $\varphi_{\sel}$ is obtained using Lemma~\ref{lem:paths-shacl}
 and $\hat\varphi$ is obtained using Lemma~\ref{lem:shapes-shacl}. 
 \end{proof}
\todo[inline]{Cem: suggestion last line of proof of \Cref{lem:schemas-shacl},

For $\varphi'$ we take $\lnot \varphi_{\sel} \lor \hat\varphi$ where $\varphi_{\sel}$ is obtained using Lemma~\ref{lem:paths-shacl} {\color{red} on $\sel$}
 and $\hat\varphi$ is obtained using Lemma~\ref{lem:shapes-shacl} {\color{red} on $\varphi$}. 

}

\subsection{Translation to ShEx}

\begin{lemma}
\label{lem:contents-shex}
For each open content type $\ntype$ there is a ShEx shape $\varphi_{\ntype}$ such that  $\graph, v\models \varphi_{\ntype}$ iff $\rho(v) \in \sem{\ntype}$ for all  $\graph=(E,\rho)$ and $v\in\nodes(\graph)$.
\end{lemma}

\begin{proof}
For the content type $\top$ the corresponding ShEx shape is $\shexneigh{\shexallte}$.

% For a content type of the form \[\closedRT{k_1:\vtype_1}\tAnd \dots \tAnd \closedRT{k_m :\vtype_m}\tAnd \top\,,\] where $k_1, \dots, k_m$ are pairwise different, 
% the corresponding ShEx shape is \[\shexneigh{k_1. \test(\vtype_1)  \shexeach \dots \shexeach k_m.\test(\vtype_m)\shexeach \shexallte}\,.\] If there are repetitions among $k_1, \dots, k_m$, we need to collect assertions about each key in a single atomic triple expression.  For example, for \[\closedRT{k_1:\vtype_1}\tAnd \closedRT{k_1:\vtype'_1}\tAnd \closedRT{k_2 :\vtype_2}\tAnd \top\,,\] we take 
% \[\shexneigh{k_1. (\test(\vtype_1)\land\test(\vtype'_1))  \shexeach k_m.\test(\vtype_2)\shexeach\shexallte}\,.\]

For a content type of the form \[\closedRT{k_1:\vtype_1}\tAnd \dots \tAnd \closedRT{k_m :\vtype_m}\tAnd \top\,,\] 
the corresponding ShEx shape is \[\shexneigh{k_1. \test(\vtype_1)\shexeach \shexallte} \land  \dots \land \shexneigh{k_m.\test(\vtype_m)\shexeach \shexallte}\,.\] 

Finally, every other  open content type can be expressed as 
\[(\ntype_1 \tOr \dots \tOr \ntype_\ell)\tAnd \top\,,\]
where each $\ntype_i$ has the form $\closedRT{k_1:\vtype_1}\tAnd \dots \tAnd \closedRT{k_m :\vtype_m}$ for some $m$.
The corresponding ShEx shape is 
\[\varphi_1 \lor \dots \lor \varphi_\ell\,,\]
where $\varphi_i$ is the ShEx shape corresponding to the content type $\ntype_i \tAnd \top$.
\end{proof}

\begin{lemma}
\label{lem:filter-shex} For each filter $\pexpr_0$ there is a ShEx shape $\varphi_{\pexpr_0}$ such that  $\graph, v\models \varphi_{\pexpr_0}$ iff $(v,v) \in \sem{\pexpr_0}^\graph$ for all  $\graph$ and $v\in\nodes(\graph)$. 
\end{lemma}

\begin{proof} By Lemma~\ref{lem:contents-shex}, the claim holds for $\pi_0=\ntype\tAnd \top$. 
For $\{k:c\}$ the corresponding ShEx shape is $\shexneigh{k.\hasvalue(c);\shexallte}$. Because ShEx shapes are closed under negation, the claim also holds for $\lnot\{k:c\}$ and $\lnot(\ntype\tAnd \top)$.  Finally, concatenations of filters correspond to conjunctions of shapes, so the claim follows because ShEx shapes are closed under conjunction. 
\end{proof}

\begin{lemma}
\label{lem:paths-shex}
For each common shape of the form $\exists\,\pexpr$ there is a ShEx shape $\varphi_{\exists\pexpr}$ such that $\graph, v\models \varphi_{\exists\pexpr}$ iff $\graph,v\models \exists\, \pexpr$ for all $\graph$ and $v\in\nodes(\graph) \cup \values(\graph)$. 
\end{lemma}

\begin{proof}
Let us first look at common shapes of the form $\exists\,\pexpr$ where $\pexpr$ is a concatenation of filters and atomic path expressions of the form  $p$, $p^{-}$, $k$, or $k^{-}$. Without loss of generality we can assume that the concatenation ends with a filter or with $k$.
We proceed by induction on the length of the concatenation. The base cases are $\exists \pi_0$ and $\exists k$, which correspond to $\varphi_{\pi_0}$ (Lemma~\ref{lem:filter-shex}) and $\shexneigh{k.\shextop \shexeach \shexallte} $, respectively.
For $\exists\, \pexpr_0\cdot\pexpr$ we can take $\varphi_{\pexpr_0} \land \varphi_{\exists \pi}$. 
For $\exists\, p\cdot\pexpr$ we can take $\shexneigh{p.\varphi_{\exists \pi} \shexeach \shexallte}$, and similarly for $\exists\, p^{-}\cdot\pexpr$ and $\exists\, k^{-}\cdot\pexpr$. 

The general case follows because ShEx shapes are closed under union. Indeed, because our PG-path expressions are star-free, we can assume without loss of generality that in each common shape of the form $\exists\, \pexpr$, the PG-path expression $\ppexpr$ underlying $\pexpr$ is a union of concatenations of filters and atomic path expressions of the form $p$ or $p^{-}$. Then, for \[\exists\, k^{-}\cdot(\pexpr^1 \cup \dots \cup\pexpr^m)\cdot k'\] we can take \[\varphi_{\exists k^{-}\cdot\pexpr^1\cdot k'} \lor \dots \lor \varphi_{\exists k^{-}\cdot\pexpr^m\cdot k}\,.\] Simiarly for $\exists\, k^{-}\cdot(\pexpr^1 \cup \dots \cup\pexpr^m)$, $\exists\, (\pexpr^1 \cup \dots \cup\pexpr^m)\cdot k'$, and $\exists\, (\pexpr^1 \cup \dots \cup\pexpr^m)$.
\end{proof}

\begin{lemma}
\label{lem:shapes-shex}
For each common shape $\varphi$ there is a ShEx shape $\hat \varphi$ such that  $\graph, v\models \varphi$ iff $\graph,v\models \hat \varphi$ for all $\graph$ and $v\in\nodes(\graph)  \cup \values(\graph)$. 
\end{lemma}

\begin{proof} 
Because ShEx shapes are closed under conjunction, it suffices to prove the claim for the atomic common shapes of the forms $\exists\, \pexpr$, $\exists^{\leq n}\pexpr_1$, $\exists^{\geq n}\pexpr_1$, and $\exists \ntype \land \not\exists\lnot P$. The first case follows from Lemma~\ref{lem:paths-shex}. 

Let us look at common shapes of the form $\exists^{\geq n}\pexpr_1$. If $n=0$ we can simply take $\shextop$. Suppose $n>0$. Then, for \[\exists^{\geq n} \pexpr_0\cdot p\cdot \pexpr'_0\] we can take 
\[\varphi_{\pexpr_0} \land \big\{\big(p. \varphi_{\pexpr'_0}\big)^{n}\shexeach \shexallte\big\},\] and similarly for $\exists^{\geq n}\pexpr_0\cdot p^{-}\cdot \pexpr'_0$, $\exists^{\geq n}\pexpr_0\cdot k^{-}\cdot \pexpr'_0$, and $\exists^{\geq n}\pexpr_0\cdot k$ (using $\shextop$ instead of $\varphi_{\pexpr'_0}$).
 
Next, we consider common shapes of the form 
$\exists^{\leq n}\pexpr_1$. For \[\exists^{\leq n} \pexpr_0\cdot p\cdot \pexpr'_0\] we can take \[\lnot\varphi_{\pexpr_0} \lor \lnot 
\big\{\big(p.\varphi_{\pexpr'_0}\big)^{n+1}\shexeach \shexallte\big\}\] and similarly for $\exists^{\leq n}\pexpr_0\cdot p^{-}\cdot \pexpr'_0$, $\exists^{\leq n}\pexpr_0\cdot k^{-}\cdot \pexpr'_0$, and $\exists^{\leq n}\pexpr_0\cdot k$ (again, using $\shextop$ instead of $\varphi_{\pexpr'_0}$). 

Before we move on, let us introduce a bit of syntactic sugar. For a set $Q = \{q_1, q_2, \dots, q_n\}\subseteq \Predicates \cup \Keys$ we write $Q^{*}$ for the triple expression
$\big(q_1.\shextop \shexone q_2.\shextop \shexone \dots \shexone q_n.\shextop\big)^{*}$. 

We are now ready to consider a common shape of the form $\exists\, \ntype \land \not\exists\,\lnot P$. Suppose first that 
\[\ntype = \closedRT{}\,.\] 
Then, the corresponding ShEx shape is simply 
\[\shexneigh{P^{*} \shexeach \big(\lnot\emptyset^{-}\big)^{*}}\,.\] 
Next, suppose that \[\ntype = \closedRT{k_1:\vtype_1}\tAnd \dots \tAnd \closedRT{k_m :\vtype_m}\,.\] 
Then, the corresponding ShEx shape is 
\[
\varphi_{\ntype\tAnd\top} \land \shexneigh{\{k_1, \dots,k_m\}^{*}\shexeach P^{*} \shexeach \big(\lnot\emptyset^{-}\big)^{*}} \,,\]
where $\varphi_{\ntype\tAnd\top}$ is obtained from Lemma~\ref{lem:contents-shex}.
In general, as in Lemma~\ref{lem:contents-shex}, we can assume that 
\[\ntype = \ntype_1 \tOr \dots \tOr \ntype_m\] where each $\ntype_i$ is of one of the two forms considered above. The corresponding ShEx shape is then 
\[\varphi_1 \lor \dots \lor \varphi_m\] where $\varphi_i$ is the ShEx shape corresponding to $\exists \ntype_i \land \not\exists\lnot P$, obtained as described above. 
\end{proof}

\begin{lemma}
\label{lem:schemas-shex}
For every common schema there is an equivalent ShEx schema.  
\end{lemma}

\begin{proof} 
Let $\schema$ be a common schema. We obtain an equivalent ShEx schema $\schema'$ by translating each  $(\sel, \varphi) \in \schema$ to $(\sel', \varphi')$ such that for all $\graph$ and $v\in \Nodes\cup\Values$, 

\begin{center}
 $\graph, v \models \sel$ implies $\graph, v \models \varphi$ 
 
 iff 
 
 $\graph, v \models \sel'$ implies $\graph, v \models \varphi'$.
 \end{center}
 Recall that $\sel$ is a common shape of one of the following forms:
\[ 
\exists\, k \,,\;
\exists\, p \cdot \pexpr\,, \;
\exists\, p^{-}\cdot \pexpr \,,\;
\exists\, \{k:c\}\cdot \pexpr \,, \;
\exists\, \big(\{k:\vtype\}\tAnd\top\big)\cdot \pexpr \,, \;
\exists\, k^{-}\cdot \pexpr\,.
\]
If $\sel$ is of the form  \[\exists\, k\,,\quad \exists\, \{k:c\}\cdot \pi\,,\quad \text{or}\quad \exists\, \big(\{k:\vtype\}\tAnd \top\big)\cdot \pi\,,\] for $\sel'$ we take $\shexneigh{ k.\!\shextop \shexeach \shexallte}$. In the remaining cases, we take, respectively, 
\[ 
\shexneigh{ p.\!\shextop \shexeach \shexallte},\quad
\shexneigh{ p^{-}\!.\!\shextop \shexeach \shexallte},\quad
\shexneigh{ k^{-}\!.\!\shextop \shexeach \shexallte}.
\]
For $\varphi'$ we take $\lnot \varphi_{\sel} \lor \hat\varphi$ where $\varphi_{\sel}$ is obtained using Lemma~\ref{lem:paths-shex}
 and $\hat\varphi$ is obtained using Lemma~\ref{lem:shapes-shex}. 
 \end{proof}

\end{document}